\documentclass[12pt,fleqn]{article}
\pdfoutput=1 
\usepackage{amsthm,amsfonts,amsmath,amssymb,amsxtra,graphicx,url,color,float}
\usepackage[margin=1in,centering]{geometry}
\usepackage{setspace,comment,url,enumerate}
\usepackage[ansinew]{inputenc}
\usepackage[sans]{dsfont}
\usepackage[T1]{fontenc}
\usepackage[dvips]{rotating}
\usepackage{mathrsfs}
\usepackage{bbm}
\usepackage{natbib}
\usepackage{enumerate}
\usepackage{xcolor, colortbl}
\usepackage{subcaption}
\usepackage{url}
\usepackage{makecell}
\usepackage{booktabs}

\theoremstyle{plain}
\newtheorem{theorem}{Theorem}
\newtheorem{proposition}{Proposition}

\theoremstyle{remark}

\theoremstyle{definition}
\newtheorem{definition}{Def\/inition}

\usepackage[bottom]{footmisc}
\usepackage{graphicx}
\usepackage{placeins}
\usepackage{tikz}

\newcolumntype{P}[1]{>{\centering\arraybackslash}p{#1}}

\definecolor{light-gray}{gray}{0.88}
\definecolor{dark-gray}{gray}{0.7}

\allowdisplaybreaks

\begin{document}\onehalfspacing

\begin{titlepage}

\title{\onehalfspacing\textbf{Time is Knowledge:\\ Response Times in Binary Response Models}}

\author{\text{Jean-Michel Benkert, Shuo Liu and Nick Netzer}\thanks{Benkert: Department of Economics, University of Bern, jean-michel.benkert@unibe.ch. Liu: Guanghua School of Management, Peking University, and Division of Social Science, New York University Abu Dhabi, shuo.liu@gsm.pku.edu.cn. Netzer: Department of Economics, University of Zurich, nick.netzer@econ.uzh.ch. We are grateful for comments by Peter Andre, Ernst Fehr, Lukas Hensel, Blaise Melly, Lasse Mononen, Costanza Naguib, Andrew Oswald, and Martin Vaeth. We also thank seminar participants at the California Institute of Technology,  Central
University of Finance and Economics, CUHK-Shenzhen, Hong Kong Baptist University, HU Berlin, NYU Abu Dhabi, Peking University, Shanghai Jiao Tong University, Shanghai University of Finance and Economics, Xiamen University, the Universities of Bamberg, Basel, Bern, Edinburgh, Essex, Freiburg, Hamburg, and Hong Kong, BSE Summer Forum 2025, CESifo Area Conference on Behavioral Economics 2025, Verein f\"ur Socialpolitik Theoretischer Ausschuss 2025, and FUR 2026. An earlier version of this paper was circulated under the title ``Time is Knowledge: What Response Times Reveal.'' Shuo Liu acknowledges financial support from the National Natural Science Foundation of China (grants 72322006 and 72192844).}}

\date{This version: October 2026 \\ First version: August 2024}


\maketitle\thispagestyle{empty}

\begin{abstract}
Response times contain information about economically relevant but unobserved variables like willingness to pay, preference intensity, quality, or happiness. We provide a general characterization of the properties of latent variables that can be detected using response time data. Our theoretical framework unifies and generalizes existing results in the literature and gives rise to many new applications. We illustrate the novel insights that the method can deliver through three empirical applications: identifying an optimal nudge, testing decreasing marginal happiness of income, and predicting treatment heterogeneity.
\end{abstract}

\vspace*{0.5cm}


\vspace*{.1cm}


\vspace*{.3cm}

\end{titlepage}

\section{Introduction}\label{sec.intro}

Traditionally, economists have ignored choice process data like response times. Only recently has the literature realized that response times contain valuable information about unobserved variables like willingness to pay \citep{krajbich2012attentional,cotet2025deliberation}, preference intensity \citep{Chabrisetal09,AFN18,AG22}, quality assessments \citep{card23}, or subjective happiness \citep{liunetzer23}.

In this paper, we take a systematic approach to understanding how response time data can be exploited. We study the identification of distributional properties of latent variables in canonical binary response models, a framework extensively used in economics and applicable to all the settings described above. Our approach builds on the observation that decisions are faster when the absolute value of the latent variable is larger. The observable response times are therefore informative about the unobservable latent variable. As our main theoretical result, we provide a full characterization of the distributional properties of latent variables that can be detected using response time data. Our characterization highlights how different assumptions about the relationship between the latent variable and response times determine what properties can be identified. 

To illustrate the usefulness of the method, we study three empirical applications where a diverse range of distributional properties matters. First, for nudging, the true preference distribution must be identified from choices that are distorted by frames. Second, in the context of happiness research, the curvature of the income-happiness relation is of particular interest. Third, predicting heterogeneous treatment effects requires detecting systematic group differences in latent distributions. In all three cases, we show that conventional methods of analysis often suffer from identification problems, but that the use of response time data with our method offers a viable solution.

The central assumption underlying our analysis is the so-called \textit{chronometric function} that associates each choice with a response time. This function is monotone, in the sense that a larger absolute value of the latent variable generates a faster decision, possibly after controlling for individual heterogeneity. We take the chronometric function as given, because a large literature has studied its foundations and provided evidence in favor of it. From a theoretical perspective, a chronometric relationship emerges naturally in evidence-accumulation models \cite[e.g.][]{Chabrisetal09,fudenberg2018speed,card23}, where a stronger stimulus generates faster decisions. The empirical evidence for a monotone chronometric function in the laboratory is vast. Among many others, \cite{kellogg1931}, \cite{moyerbayer1976}, and \cite{palmer2005} document the effect for choice situations with an objective stimulus, and \cite{Moffatt05}, \cite{Chabrisetal09}, \cite{konovalov2019revealed}, and \cite{AG22} for subjective value-based environments. Field evidence is also emerging. \citet{card23} document that editorial decisions take longer when the submitted paper's quality implies a closer decision. Using eBay data on bargaining behavior, \citet{cotet2025deliberation} show that a seller's response time to an offer systematically depends on its perceived value. In the context of an online survey, \citet{liunetzer23} demonstrate that faster responses are associated with a stronger sense of approval for the selected answer.

To understand our main insight, consider a decision-making environment where one or multiple agents choose between two options. Choice is determined by the realization $x$ of an unobservable latent variable, with $x \leq 0$ generating choice of option $i=0$ and $x>0$ generating choice of option $i=1$. An analyst observing the choices wants to learn about the underlying cumulative distribution function $G$ of the latent variable. Unfortunately, the only property of $G$ that is identified without additional assumptions or data is its value at zero, $G(0)$, which is given by the observed probability or frequency of choosing $i=0$. This information is extremely limited and does, for example, not imply anything about the mean of $G$ without additional distributional assumptions. Now suppose that the speed of the decision is given by $c(|x|)$ for a strictly decreasing chronometric function $c$, assumed here to be identical for both choice options and all subjects just for ease of exposition. Since a choice of $i=0$ arises at time $t$ or earlier if $x$ is sufficiently far below zero, where ``sufficiently far'' is determined by the chronometric function, the observed probability or frequency of choosing $i=0$ at time $t$ or earlier pins down the value of $G(-c^{-1}(t))$, and analogously for choices of $i=1$. Observing the joint distribution of responses and response times therefore allows the analyst to identify a composition of the distribution $G$ and the (inverse of the) chronometric function $c$.
 
Consequently, if the analyst had perfect knowledge of the chronometric function linking values to response times, she could recover the entire latent distribution. Such detailed knowledge is not necessary for inferring only specific distributional properties. Our main result characterizes which properties (or their violations) can be detected under which assumptions on the chronometric function. For example, detecting properties that are preserved under monotone transformations requires only knowledge of monotonicity of the chronometric function. If we further restrict attention to chronometric functions that are identical for both choice options, as in the above illustration, we can detect properties that are preserved under symmetric monotone transformations, and analogously for other classes of transformations.

Our result also provides a simple recipe how to detect or reject any property of interest. It involves constructing a candidate distribution based on the empirical data using a representative chronometric function that the analyst deems possible. If the property of interest holds for this candidate distribution, it must hold for all chronometric functions that can be obtained from the representative one using any transformation under which the property is preserved. We discuss an extension that combines this approach with distributional assumptions. Additionally, we show how to incorporate individual heterogeneity and noise into the analysis, and we provide necessary and sufficient conditions for the rationalizability of response time data.

The existing econometrics literature has studied identification of binary response models through assumptions on the distribution of the latent variable and exogenous variation of observables \citep[e.g.,][]{Manski88, matzkin1992nonparametric}. Some of the assumptions required to achieve identification have been criticized \citep[e.g.,][]{haile2008empirical, BL19}. Our method based on response times is complementary to that literature and allows us to avoid some of the controversial assumptions. In the context of stochastic choice theory, \citet{AFN18} have shown that response time data can be used to obtain revealed preferences without making distributional assumptions about the random utility component and to improve out-of-sample predictions.  In the context of happiness surveys, \cite{liunetzer23} have shown that response times can be used to test distributional assumptions of conventional econometric models. Several of the results in these two papers follow as corollaries from our characterization here and can be generalized.

Our approach lends itself to a wide range of applications. We focus on three applications in the paper. For each of them, we use our general theoretical result to derive conditions that allow us to detect or reject distributional properties which are of interest in that specific application. We then apply these conditions to existing data and evaluate their performance.

Our first application concerns the problem of selecting an optimal nudge \citep{thalersun08,benkert18}. We assume that a choice architect must select between two frames, for example defaults, each distorting choice towards one of the two options. The distorted choices are not directly informative about the shares of agents who truly prefer each option. However, under the assumption that the truly indifferent agent is distorted symmetrically under the two frames, response time data allow us to detect the true preference shares and to select the optimal nudge. We apply this result to data from \citet{serragarcia23}, where the existence of a neutral frame allows us to validate the method. Our response time-based prediction is strikingly accurate and outperforms alternative methods proposed in the literature, both in the original data and in synthetic resamples that we create to investigate the variability of the predictions. In addition to showcasing our method's scope for addressing important real-world problems, these results also constitute another successful test of the chronometric function.

In our second application, we test the hypothesis of decreasing marginal happiness of income, a principle central to redistributive policies. \citet{Oswald08} and \citet{KaiserOswald22} question the empirical foundation of this principle by arguing that an observed concave relationship between income and self-reported happiness may result from a concave reporting function rather than from decreasing marginal happiness. Conventional approaches used in the happiness literature are insufficient to establish the principle \citep{BL19}. We show that the principle becomes testable with the response time-based method, using conditions for detecting the ranking of the means of distributions. Our analysis of survey data by \citet{liunetzer23} reveals that the hypothesis of decreasing marginal happiness cannot be rejected (while the hypothesis of increasing marginal happiness is clearly rejected). Again, the response time-based method contributes to answering a major open question, which is at the foundation of debates surrounding justice and inequality.

Our third application concerns heterogeneous treatment effects \cite[e.g.,][]{Manski2004Statistical}. Heterogeneity is relevant for policymakers who need to decide which subpopulation to target with a treatment, for firms who want to set different prices in different markets, and for experimental economists who are interested in bounds on treatment effects. The existing literature typically requires post-treatment data to predict which subgroup of a population will respond more strongly to a treatment \citep[e.g.,][]{Athey2016PNAS}. We show that heterogeneity can be predicted even pre-treatment, because subgroups with a larger near-indifferent mass of agents will respond more strongly to a treatment and this property can be detected using response time data. Using the data of \citet{Krefeld2024PNAS}, we once more validate our approach by showing that we correctly predict the overwhelming majority of subgroup comparisons.

The three applications are merely illustrations of the scope of the method and the value of using response times in economics. Going beyond these applications, our approach could be used to detect polarization of political attitudes from ordinal survey questions \citep{lelkes2016,vaeth23}, infer properties of demand functions that are important for pricing from purchase decisions at a single price \citep{johnson2006simple}, and uncover correlations that are not directly observable when subjects' responses are uninformative because they conflict with social norms \citep{coffman2017size,konovalov2019revealed}. As a template for future work, Appendix \ref{sec.theory.app} contains a collection of distributional properties that play a role in these contexts and, for each of them, derives both the identifying response time conditions and the required assumptions on the chronometric function.

We conclude the introduction with guidelines for applied work. Our framework helps the analyst determine which restrictions on the chronometric function are needed for detecting a distributional property of interest. In assessing whether these restrictions are reasonable for the problem at hand, it is useful to distinguish three factors beyond latent intensity that influence response times. First, individual heterogeneity and noise affect how fast agents respond. We will show that heterogeneity can often be addressed by normalization using response time from a baseline question, and we provide conditions under which the relevant tests remain valid in the presence of noise, but some problems cannot be solved that way. For example, normalization cannot address systematic differences in how much two groups of survey participants are confused by a question if they respond equally fast to the baseline question. Second, not only intensity may vary across problems but also complexity. A famous example, discussed by \cite{Tversky72} as a critique of conventional stochastic choice models, involves one choice problem with a dominance relation (trip to Paris vs.\ trip to Paris plus one dollar) and one without (trip to Paris vs.\ trip to Rome). A faster and less random response in the first problem primarily indicates greater simplicity rather than greater preference intensity. More generally, \cite{goncalves24} highlights that response times can become non-monotone if a decision-maker observes and adjusts to the complexity of a problem before starting the decision process. Third, different decision-making procedures can generate differences in response times, as in the analysis of instinctive versus contemplative choices by \cite{rubinstein2007instinctive,rubinstein2013response,rubinstein2016typology}. These considerations can guide the assessment of the chronometric restrictions required for a particular application. A chronometric effect can be expected if latent intensity is the main source of variation across subjects or problems. Identical chronometric effects across subjects or choice options become more plausible when there are fewer observable differences between them and less scope for adjustment. Where a benchmark is available, as in our nudging application, comparing the resulting predictions with that benchmark provides an additional test of the joint assumptions used for detection.

The paper is organized as follows. Section \ref{sec.theory} introduces the framework and derives our main theoretical result, followed by extensions. Section \ref{sec.app} presents the three applications. Section \ref{sec.concl} concludes. Supporting materials are provided in the appendices.

\section{Theory}\label{sec.theory}

In this section, we develop our general theoretical framework and present our main result, which shows how and under which conditions distributional properties can be detected using response time data. 

\subsection{Binary Response Model}

We first introduce the binary response model \cite[e.g.,][]{Manski88}. There is a random variable $\tilde{x}$ with values $x \in \mathbb{R}$ that induce binary responses by comparison with a decision threshold. We normalize the threshold to zero without loss of generality. Thus, the response is $i=0$ if $\tilde{x}$ takes a value $x \leq 0$ and $i=1$ if $\tilde{x}$ takes a value $x > 0$. We describe the distribution of the latent variable $\tilde{x}$ by a cumulative distribution function (cdf) $G$, which we assume to be continuous. It follows that the probabilities of the two responses are $p^0=G(0)$ and $p^1=1-G(0)$.

The model has different applications and interpretations. The latent variable could be a random utility difference $\tilde{x}= [u(1) + \tilde{\epsilon}(1)] - [u(0) + \tilde{\epsilon}(0)]$ between two options, inducing stochastic choices of a single agent \citep[as in][]{AFN18}. In another application, $\tilde{x}$ could describe the difference between the willingness to pay $\tilde{v}$ for a product among consumers and the product price $p$, $\tilde{x}=\tilde{v}-p$, inducing the demand for the product at that price \citep[in the spirit of][]{cotet2025deliberation}. In yet another application, $\tilde{x}$ could capture the distribution of happiness in a population of agents, inducing frequencies of responses to a binary survey question about happiness \citep[as in][]{liunetzer23}. The same logic applies to other survey questions, where the responses could be driven by a distribution of political or moral attitudes or preference parameters in the population. The latent variable could also describe the random quality of papers that are submitted to a journal, inducing the editor's decision to accept or reject \cite[as in][]{card23}.

We now follow \cite{AFN18} and \cite{liunetzer23} and assume that the realized value $x$ of $\tilde{x}$ not only determines the response but also the response time, with larger absolute values implying faster responses, in line with the chronometric effect.\footnote{There exists a literature in psychology that introduces response times into response models in a similar way \citep[e.g.,][]{takane83,ferrando07a,ferrando07b,ranger13}. However, that literature is interested in improving parameter estimation using response time data, rather than detecting distributional properties, and therefore imposes parametric functional forms and makes the conventional distributional assumptions.} Formally, we denote by $c: \mathbb{R} \rightarrow [\underline{t},\overline{t}]$ the chronometric function, where $0 \leq \underline{t} < \overline{t} < \infty$. The function $c$ maps each realized value $x$ into a response time $c(x)$. It is assumed to be continuous, strictly increasing on $\mathbb{R}_-$ and strictly decreasing on $\mathbb{R}_+$ whenever $c(x)>\underline{t}$, and to satisfy $c(0)=\overline{t}$ and $\lim_{x \rightarrow -\infty} c(x)=\lim_{x \rightarrow +\infty} c(x)=\underline{t}$. Figure \ref{fig:chronometric_function} illustrates two examples of chronometric functions that satisfy all these conditions.

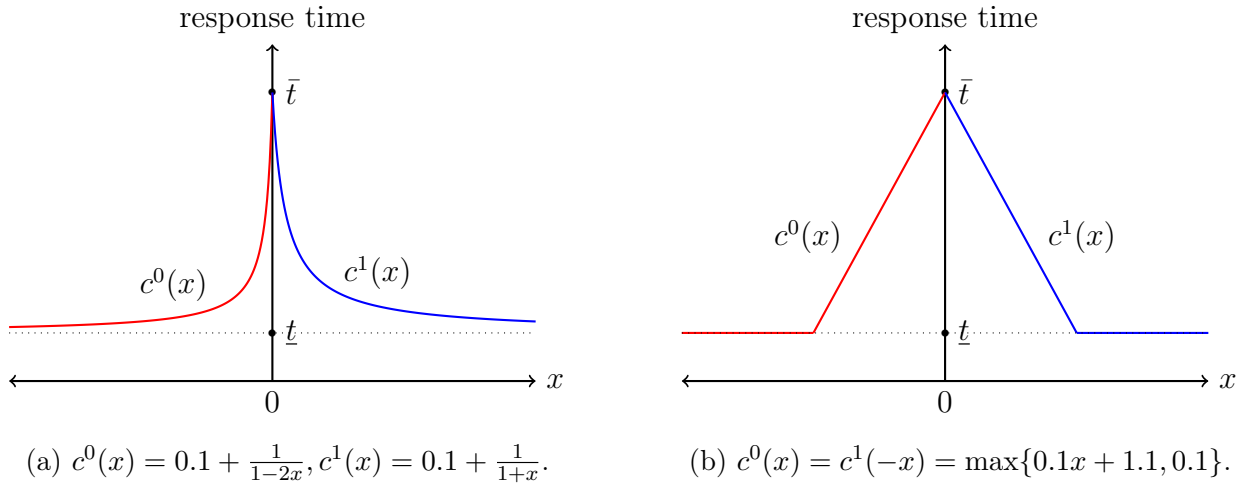
\begin{figure}
\centering
\begin{subfigure}{0.45\textwidth}
     \begin{tikzpicture}[yscale=3.2, xscale=0.175]
     \node at (0, 0.1) {\tiny $\bullet$};
     \node at (0, 1.1) {\tiny $\bullet$};
            \node[right] at (0.2, 0.1) {$\underline{t}$};
            \node[right] at (0.2, 1.1) {$\bar{t}$};
            \draw[thick] (-0.2, 0.1) -- (0.2, 0.1); 
            \draw[thick] (-0.2, 1.1) -- (0.2, 1.1); 
            \node[below] at (0, -0.1) {$0$};
            \node[right] at (20, -0.1) {$x$};
            \node[above] at (0, 1.3) {response time};
            \draw [thick, <->](-20,-0.1) -- (20, -0.1);
            \draw [thick, ->](0,-0.1) -- (0, 1.3);
            \draw[thick, red, domain=-20:0, samples=200] plot (\x, {0.1 + 1/(-2*\x +1)});
            \draw[thick, blue, domain=0:20, samples=200] plot (\x, {0.1 + 1/(\x +1)});
            \draw[dotted, domain = -20:20] plot (\x, 0.1);
            \node at (-7.5, 0.3) {$c^0(x)$};
            \node at (8, 0.35) {$c^1(x)$};
            \end{tikzpicture}
    \caption{$c^0(x)=0.1+\frac{1}{1-2x}, c^1(x)=0.1+\frac{1}{1+x}$.}
    \label{fig:asymptotic}
\end{subfigure}
\hfill
\begin{subfigure}{0.45\textwidth}
     \begin{tikzpicture}[yscale=3.2, xscale=0.175]
           \node at (0, 0.1) {\tiny $\bullet$};
     \node at (0, 1.1) {\tiny $\bullet$};
            \node[right] at (0.2, 0.1) {$\underline{t}$};
            \node[right] at (0.2, 1.1) {$\bar{t}$};
            \draw[thick] (-0.2, 0.1) -- (0.2, 0.1); 
            \draw[thick] (-0.2, 1.1) -- (0.2, 1.1); 
            \node[below] at (0, -0.1) {$0$};
            \node[right] at (20, -0.1) {$x$};
            \node[above] at (0, 1.3) {response time};
            \draw [thick, <->](-20,-0.1) -- (20, -0.1);
            \draw [thick, ->](0,-0.1) -- (0, 1.3);
            \draw[thick, red, domain=-10:0, samples=200] plot (\x, {(1/10)*\x+1.1});
            \draw[thick, blue, domain=0:10, samples=200] plot (\x, {-(1/10)*\x+1.1});
            \draw[thick, blue, domain = 10:20] plot (\x, 0.1);
            \draw[thick, red, domain = -10:-20] plot (\x, 0.1);
            \draw[dotted, domain = -20:20] plot (\x, 0.1);
            \node[left] at (-7, 0.5) {$c^0(x)$};
            \node[right] at (7, 0.5) {$c^1(x)$};
            \end{tikzpicture}
    \caption{$c^0(x)=c^1(-x)=\max\{0.1x+1.1, 0.1\}$.}
    \label{fig:finite}
\end{subfigure}
\vspace{2mm}
   \caption{Examples of chronometric functions.}
   \vspace{3mm}
   \begin{minipage}{\textwidth}
   \small\textit{Notes:} The left panel in the figure depicts an asymmetric chronometric function that asymptotically approaches the fastest response time $\underline{t}$, while the right panel shows a symmetric one that  attains $\underline{t}$ at finite absolute values of the latent variable. In both panels, the red and blue curves correspond to the restrictions of $c$ to $\mathbb{R}_{-}$ and $\mathbb{R}_+$, respectively.
  \end{minipage}
\label{fig:chronometric_function}
  \end{figure}

The restriction of $c$ to $x \in \mathbb{R}_-$ is denoted $c^0$. This function $c^0$ has a well-defined inverse $(c^0)^{-1}: (\underline{t},\overline{t}] \rightarrow \mathbb{R}_-$ that is continuous and strictly increasing. We extend it to $\underline{t}$ by setting $(c^0)^{-1}(\underline{t})=-\infty$ if $c(x)>\underline{t}$ for all $x \in \mathbb{R}_-$ and $(c^0)^{-1}(\underline{t})= \max \{x \in \mathbb{R}_- \mid c(x)=\underline{t} \}$ otherwise. Analogously, the restriction of $c$ to $x \in \mathbb{R}_{+}$ is denoted $c^1$, with a continuous and strictly decreasing inverse $(c^1)^{-1}: (\underline{t},\overline{t}] \rightarrow \mathbb{R}_{+}$, which we extend to $\underline{t}$ by setting $(c^1)^{-1}(\underline{t})=+\infty$ or $(c^1)^{-1}(\underline{t})= \min \{x \in \mathbb{R}_{+} \mid c(x)=\underline{t} \}$, as appropriate.

In addition to response probabilities, the model $(G,c)$ induces distributions of response times. We denote by $F^i$ the cdf of response times conditional on a response of $i=0,1$. Since a response $i=0$ at time $t$ or earlier arises if $x \leq (c^0)^{-1}(t)$, we obtain that
\begin{equation}\label{gen0} p^0 F^0(t)=G((c^0)^{-1}(t))\end{equation}
for all $t \in [\underline{t},\overline{t}]$, where we use the convention $G(-\infty)=0$. Analogously, a response $i=1$ at time $t$ or earlier arises if $x \geq (c^1)^{-1}(t)$, so that
\begin{equation} \label{gen1} p^1 F^1(t)=1 - G((c^1)^{-1}(t))\end{equation}
for all $t \in [\underline{t},\overline{t}]$, where we use $G(+\infty)=1$. The induced response-time cdfs $F^i$ are continuous on $[\underline{t},\overline{t}]$ and satisfy $F^i(\overline{t})=1$.

In summary, the binary response model $(G, c)$ induces the data $(p,F)=(p^0,p^1,F^0,F^1)$ according to (\ref{gen0}) and (\ref{gen1}).  

\subsection{Detecting Properties}

We now ask what an analyst can learn from observed data about the underlying binary response model and, in particular, about the distribution $G$ of the latent variable. Taking observed data as given, different binary response models could have generated those data, so inference about the model is not straightforward. This is true especially if the analyst is not willing to make potentially strong assumptions about the form of the chronometric function or the latent distribution. We ask whether there are some properties that all models satisfy which are consistent with the data.

Consider a profile of data $(p_j,F_j)_{j} =(p^0_j,p^1_j,F^0_j,F^1_j)_j$ indexed by $j \in J$. The set $J$ could be a singleton, e.g., when studying the choices of a single agent between two options or the responses of a single group of agents to one binary survey question. In this case, we typically omit the index $j$. If we observe the choices of a single agent for multiple pairs of options, then the set $J$ would describe the different binary choice problems. Alternatively, the index could capture different agents or combinations of agents and choice problems. In a survey application, the index could describe different questions or different demographic groups. Since $J$ is not necessarily finite, it could also model income levels $j \in J\subseteq\mathbb{R}_+$ of survey participants. In the other applications discussed earlier, the index could describe different journal editors or different prices at which market demand is observed. We assume that each $F^i_j$ is continuous and satisfies $F^i_j(\underline{t})=0$ and $F^i_j(\overline{t})=1$.\footnote{Even though our approach is very general, it does not allow describing correlations between choices or response times across different indices $j$. It also does not allow mass points in the response time distributions. We leave such extensions for future work.}

Denote by $\mathscr{G}$ the set of all possible profiles $(G_j)_j$ of cdfs, where each individual $G_j$ satisfies our assumptions from the previous subsection. Similarly, denote by $\mathscr{C}$ the set of all possible profiles $(c_j)_j$ of chronometric functions that individually satisfy our previous assumptions. For our main result, we leave $\mathscr{G}$ unrestricted but allow for a possibly restricted set $\mathscr{C}^* \subseteq \mathscr{C}$ of admissible chronometric functions. For example, each profile $(c_j)_j \in \mathscr{C}^*$ may have to satisfy that all functions are symmetric across responses (i.e., $c_j^0(-x)=c_j^1(x)$ for all $x \in \mathbb{R}_+$). This restriction embodies the assumption that the chronometric effect is the same for the two choice options. Another possible restriction would be that the functions are identical across indices (i.e., $c_j=c$ for all $j \in J$), reflecting the assumption that the chronometric effect is the same for all groups $j$. Other restrictions could be functional forms such as piece-wise linearity, or combinations of multiple of these assumptions.

\begin{definition}\label{defdet}
Given data $(p_j,F_j)_{j}$ and a set $\mathscr{C}^*$ of admissible chronometric functions, a property $\mathbf{P}$ of the distributions $(G_j)_j$ is \textit{detected} if all $((G_j)_j,(c_j)_j) \in \mathscr{G} \times \mathscr{C}^*$ that induce $(p_j,F_j)_{j}$---in the sense of (\ref{gen0}) and (\ref{gen1}) for all $j \in J$ ---have in common that $(G_j)_j$ satisfies $\mathbf{P}$.
\end{definition}

Detection formalizes that property $\mathbf{P}$ is not assumed by the analyst but inferred from the data, because it holds in every model that is consistent with the data and the maintained restrictions on the chronometric functions. This corresponds to a standard notion of non-parametric identification in econometrics \citep[e.g.,][]{Manski88,BL19}. The revealed preference approach in choice theory \citep{samuelson38,arrow59} embodies the same logic \citep[see, e.g.,][]{benkert18,AFN18}.\footnote{We assume here implicitly that there exists at least one $((G_j)_j,(c_j)_j) \in \mathscr{G} \times \mathscr{C}^*$ that induces the data. As we will see later, this is always the case when $\mathscr{G}$ is unrestricted. In Subsection \ref{subsec.dconstr} we will study the case with restricted $\mathscr{G}^* \subseteq \mathscr{G}$ and provide conditions for rationalizability of the data.}

The ability to detect a property $\mathbf{P}$ will depend on the extent to which $\mathbf{P}$ is invariant to transformations. As a first illustrative example, consider the asymmetry property of a single cdf $G$ defined by
\begin{equation} \label{TWTasym}
G(-x) \leq 1- G(x) \; \text{for all} \; x \in \mathbb{R}_+,
\end{equation}
which implies that the mean of $G$ is positive. This property is invariant to strictly increasing transformations $\psi$ that are symmetric around zero, as it implies $G(\psi(-x)) \leq 1- G(\psi(x))$ whenever $\psi(x) = - \psi(-x)$ for all $x \in \mathbb{R}_+$. As a second example, consider the property that $G_1$ first-order stochastically dominates $G_2$ defined by 
\begin{equation} \label{HTfosd}
G_1(x) \leq G_2(x) \; \text{for all} \; x \in \mathbb{R}.
\end{equation}
It is invariant to profiles $(\psi_1,\psi_2)$ of strictly increasing transformations that are identical for the two distributions, as it implies $G_1(\psi_1(x)) \leq G_2(\psi_2(x))$ when $\psi_1(x)=\psi_2(x)$ for all $x \in \mathbb{R}$.

Generally, let $\Psi$ be a set of profiles $(\psi_j)_j$ of functions $\psi_j: \mathbb{R} \rightarrow \mathbb{R}$ that are bijective and strictly increasing, hence continuous. The set $\Psi$ can embody various constraints, such as the restriction that each profile $(\psi_j)_j \in \Psi$ is composed only of identical and/or symmetric functions. For any $(G_j)_j \in \mathscr{G}$ and any $(\psi_j)_j \in \Psi$, the composition $(G_j \circ \psi_j)_j$ is another profile of cdfs in $\mathscr{G}$. 

\begin{definition}\label{definvar}
A property $\mathbf{P}$ of $(G_j)_j$ is \textit{invariant to transformations $\Psi$} if $(G_j \circ \psi_j)_j$ also has property $\mathbf{P}$, for all $(\psi_j)_j \in \Psi$.
\end{definition}

Classifying properties by their invariance to transformations is common in the theory of measurement \citep{measurementbook} and parallels an approach used in social choice theory to describe the measurability and interpersonal comparability of utilities that different welfare functions require \citep{daspre02}.

\subsection{Generating Chronometric Functions}

We consider sets $\mathscr{C}^*$ of chronometric functions that are generated by a representative profile $(c_j^*)_j \in \mathscr{C}$ and a set of transformations $\Psi$ as introduced above. 

\begin{definition}\label{defgen}
The pair $((c_j^*)_j,\Psi)$ \textit{generates $\mathscr{C}^*$} if for each $(c_j)_j \in \mathscr{C}^*$ there exists $(\psi_j)_j \in \Psi$ such that $(c_j)_j = (c_j^* \circ \psi_j)_j$. \end{definition}

To illustrate the concept, we discuss examples of sets $\mathscr{C}^*$ and how they can be generated. Consider the set of all profiles of chronometric functions which approach $\underline{t}$ asymptotically in the limit but never reach $\underline{t}$. This set is generated by a simple representative member $(c_j^*)_j$, for example the symmetric hyperbolic form
\begin{equation} \label{ci} c_j^*(x)=\underline{t} + \frac{1}{|x| + 1/(\overline{t}-\underline{t})},
\end{equation}
identical for each $j \in J$, together with the unrestricted set of all profiles of transformations. To see why, just note that any desired $(c_j)_j$ can be obtained from $(c_j^*)_j$ by using the transformations $(\psi_j)_j$ given by
\[\psi_j(x)=\begin{cases} (c^{*,1}_j)^{-1}(c_j(x)) & \text{if } x>0, \\ (c^{*,0}_j)^{-1}(c_j(x)) & \text{if } x\leq 0,\end{cases}\]
for each $j \in J$. Similarly, the set of all profiles of chronometric functions which have $c_j(x)=\underline{t}$ for large but finite absolute values of $x$ can be generated by a simple representative member, for example the symmetric linear form \begin{equation} \label{cf} c^*_j(x) =\begin{cases} \underline{t} & \text{if } (\overline{t}-\underline{t}) < x,\\ \overline{t} - x & \text{if } 0 < x \leq (\overline{t}-\underline{t}), \\ \overline{t} + x & \text{if } -(\overline{t}-\underline{t}) \leq x \leq 0,\\ \underline{t} & \text{if } x < -(\overline{t}-\underline{t}), \\ \end{cases}\end{equation} identical for each $j \in J$, together with the unrestricted set of all profiles of transformations.\footnote{Note that only transformations $\psi_j$ that satisfy $\psi_j(0)=0$ yield well-defined chronometric functions, while $\Psi$ could contain functions without that property. More generally, our definition does not require that all elements of $\Psi$ must be used for the construction of $\mathscr{C}^*$. An example of a set $\mathscr{C}^*$ which cannot be generated by any $((c_j^*)_j,\Psi)$ is one which contains some profiles in which $c_j$ reaches $c_j(x)=\underline{t}$ and other profiles in which $c_j(x)>\underline{t}$ throughout, for some $j \in J$. However, this can be addressed by our generalization in Subsection \ref{subsec.mutl} which allows for the generation of sets by more than one representative profile $(c_j^*)_j$.}

If we combine either (\ref{ci}) or (\ref{cf}) with the smaller set of transformations that are symmetric around zero, we can generate the respective sets of profiles of chronometric functions that are symmetric across responses. If we combine (\ref{ci}) or (\ref{cf}) only with transformations that are identical across indices, we generate only profiles of chronometric functions that are identical across indices. Different of these cases will be relevant in our applications in Section \ref{sec.app}.

\subsection{Main Result}

Given observed data $(p_j,F_j)_j$ and a representative profile $(c^*_j)_j$ of chronometric functions, we can derive the empirical distribution functions
\begin{equation}\label{defH}H_j(x) = \begin{cases} 1 - p^1_j F^1_j(c^*_j(x)) & \text{if } x>0, \\ p^0_j F^0_j(c^*_j(x)) & \text{if } x \leq 0,\end{cases}\end{equation}
for all $j \in J$. Each $H_j$ is a well-defined and continuous cdf. It would be the true cdf of the binary response model inducing $(p_j,F_j)$ if $c^*_j$ was the true chronometric function.\footnote{See \citet[Online Appendix p.\ 5]{AFN18} for an analogous construction in the context of a random utility model for a single binary choice.}

Of course, $c^*_j$ may well be different from the true chronometric function. However, the proof of our main result will show that whenever an admissible chronometric function $c_j$ can be written as $c_j = c_j^* \circ \psi_j$ for some $\psi_j$, the corresponding distribution $G_j$ can be written as $G_j = H_j \circ \psi_j$. Therefore, if $(H_j)_j$ satisfies a property that is invariant to $\Psi$, then any $(G_j)_j$ that is compatible with the data also satisfies that property, under the assumption that $((c^*_j)_j,\Psi)$ generates $\mathscr{C}^*$. This gives rise to the following result. 

\begin{theorem}\label{t1} Suppose $\mathscr{C}^*$ is generated by $((c^*_j)_j,\Psi)$. If $(H_j)_j$ defined in (\ref{defH}) satisfies a property $\mathbf{P}$ that is invariant to transformations $\Psi$, then $\mathbf{P}$ is detected.
\end{theorem}

\begin{proof} Let $((G_j)_j,(c_j)_j) \in \mathscr{G} \times \mathscr{C}^*$ be any model that induces $(p_j,F_j)_j$. Using the definition of $(H_j)_j$ and the conditions (\ref{gen0}) and (\ref{gen1}) for inducing the data, we obtain that
\[H_j(x) = \begin{cases} G_j((c^1_j)^{-1}(c^{*1}_j(x))) & \text{if } x>0, \\ G_j((c^0_j)^{-1}(c^{*0}_j(x))) & \text{if } x \leq 0,\end{cases}\]
for each $j \in J$. Since $((c^*_j)_j,\Psi)$ generates $\mathscr{C}^*$, there exists $(\psi_j)_j \in \Psi$ such that $c_j = c^*_j \circ \psi_j$, for each $j \in J$. We therefore have $c^0_j(x)=c^{*0}_j(\psi_j(x))$ for any $x \leq 0$ and $c^1_j(x)=c^{*1}_j(\psi_j(x))$ for any $x > 0$. It follows that $(c^0_j)^{-1}(t)=\psi_j^{-1}((c^{*0}_j)^{-1}(t))$ and $(c^1_j)^{-1}(t)=\psi_j^{-1}((c^{*1}_j)^{-1}(t))$ for all $t \in [\underline{t},\overline{t}]$. It then follows that
\[H_j(x) = \begin{cases} G_j(\psi^{-1}_j(x)) & \text{if } x>0, \\ G_j(\psi^{-1}_j(x) ) & \text{if }  x \leq 0.\end{cases}\]
Note that this is also true when $\bar{x}^*_j:=(c^{*0}_j)^{-1}(\underline{t})$ is finite and we consider any $x \leq \bar{x}^*_j$. In that case, $\bar{x}_j:=(c^{0}_j)^{-1}(\underline{t})=\psi^{-1}_j(\bar{x}^*_j)$ is also finite, and we have $G_j(x)=0$ for all $x \leq \bar{x}_j$ because $(G_j,c_j)$ induces data with $F^0_j(\underline{t})=0$. For all $x \leq \bar{x}^{*}_j$, we therefore have $G_j((c^0_j)^{-1}(c^{*0}_j(x)))=G_j(\bar{x}_j)=0=G_j(\psi^{-1}_j(\bar{x}^*_j))=G_j(\psi^{-1}_j(x))$. The analogous argument applies to $x \geq (c^{*1}_j)^{-1}(\underline{t})$. To summarize, we have $H_j = G_j \circ \psi^{-1}_j$, for each $j \in J$. If $(H_j)_j$ satisfies a property $\mathbf{P}$ that is invariant to transformations $\Psi$, then $(H_j \circ \psi_j)_j = (G_j \circ \psi^{-1}_j \circ \psi_j)_j = (G_j)_j$ also satisfies $\mathbf{P}$, hence $\mathbf{P}$ is detected.
\end{proof}

Theorem \ref{t1} is easy to apply in a broad range of applications. If we are interested in some distributional property, we only need to check specific empirical distributions $(H_j)_j$ that are based on a representative profile $(c_j^*)_j$ of chronometric functions. These distributions can be constructed from the observed data. If they satisfy a property that is invariant to transformations $\Psi$, then this property is detected for the entire class of chronometric functions that $((c^*_j)_j,\Psi)$ generates. As we will show below, the resulting detection conditions can often be expressed directly and transparently in terms of the observable data, without recourse to specific functions $(c^*_j)_j$.

The theorem can also be used to detect whether a property is violated. Denote by $\mathbf{\neg  P}$ the property that property $\mathbf{P}$ is not true. Our result immediately implies that if the empirical distributions $(H_j)_j$ satisfy $\mathbf{\neg  P}$ and this property is invariant to transformations $\Psi$, then $\mathbf{\neg  P}$ is detected. In words, the violation of a property is detected when $(H_j)_j$ violates the property and the violation is invariant to transformations $\Psi$. In that sense, the theorem provides a full characterization of the detectability of distributional properties using response times.\footnote{Note that $\mathbf{P}$ and $\mathbf{\neg P}$ can be invariant to the same or to different sets of transformations. We also remark that detecting $\mathbf{\neg P}$ is stronger than not detecting $\mathbf{P}$. A property is not detected if \textit{at least one} model that is compatible with the data violates the property, while detecting $\mathbf{\neg P}$ requires that \textit{all} models that are compatible with the data violate the property \cite[see the discussion in][p.\ 3303]{liunetzer23}.}

\subsection{Illustration}\label{subsectillu}

For illustration, we first show how existing results in the literature can be reproduced as immediate corollaries of Theorem \ref{t1}.

\cite{AFN18} study random utility models where $G$ describes the distribution of $\tilde{x}=\left[u(1)  + \tilde{\epsilon}(1) \right] - \left[u(0) + \tilde{\epsilon}(0)\right]$. When the errors $\tilde{\epsilon}(\cdot)$ have mean zero, the mean of $\tilde{x}$ equals $u(1)-u(0)$, and deducing the agent's ordinal preference can be rephrased as detecting the sign of the mean.\footnote{Some distributions do not have a mean. This can be dealt with either by using the procedure described in Subsection \ref{subsec.dconstr} to restrict the set of distributions to those which have a mean, or by refining the desired property, for example to ``the mean exists and is positive.'' Analogous arguments apply whenever a property is not well-defined for all possible distributions.} Consider then the asymmetry property (\ref{TWTasym}) discussed earlier, which implies that the mean is positive. It is invariant to transformations that are symmetric around zero, so detecting it requires knowledge of the chronometric function up to symmetric transformations. The most natural assumption guaranteeing this is that the chronometric function itself is symmetric, i.e., the chronometric effect is the same for the two options. \cite{AFN18} made that assumption throughout. Then, irrespective of which representative symmetric $c^*$ we use, the distribution $H$ defined in (\ref{defH}) exhibits the desired asymmetry if and only if
\begin{equation}\label{twtcond}p^0F^0(t) \leq p^1F^1(t) \text{ for all } t\in[\underline{t},\overline{t}].\end{equation}
This inequality is the same as in Theorem $1$ of \citet{AFN18}.
These authors discuss in detail that observing choice frequencies $p^0 \leq p^1$ is not sufficient for a preference $u(0) \leq u(1)$ to be revealed without additional assumptions on error distributions like in conventional logit or probit models, posing a severe identification problem for random utility models. However, if the inequality holds for all response times, as stated in (\ref{twtcond}), then a preference is robustly revealed without distributional assumptions.\footnote{\cite{AFN18} report that slightly more than $60\%$ of stochastic choices in the data of \cite{clithero18predict} robustly reveal a preference. Several other authors have by now used the same condition to answer open questions in different fields. \cite{AFG23} show that a sizable fraction of choices that violate stochastic transitivity in experiments reveal non-transitive preferences and thus cannot be explained by transitive preferences together with noise. \cite{AFFG24} show that common ratio and common consequence effects reflect preferences rather than noise, and \cite{castillo24} confirm that choice reversals in line with salience theory reflect true preference reversals.}

\cite{liunetzer23} study ordered response models to learn from survey data which of two groups $j=A,B$ is happier. A condition for ranking two groups irrespective of the cardinal scale on which happiness is measured is first-order stochastic dominance (FOSD), as defined earlier in (\ref{HTfosd}). FOSD is invariant to pairs of identical transformations, so detecting it requires knowledge of the chronometric functions up to identical transformations. The most natural assumption guaranteeing this is that the chronometric functions are identical, i.e., the chronometric effect is the same for the two groups. \cite{liunetzer23} made that assumption throughout. Then, irrespective of which representative $(c^*_A,c_B^*)$ with $c^*_A=c_B^*$ we use, the distributions $H_A$ and $H_B$ defined in (\ref{defH}) exhibit FOSD if and only if
\begin{align}\label{htcond}
p_A^0 F_A^0(t) - p_B^0 F_B^0(t) \leq 0 \leq p_A^1 F_A^1(t) - p_B^1 F_B^1(t) \text{ for all } t\in[\underline{t},\overline{t}].
\end{align}
These inequalities coincide with $(i)$ and $(ii)$ of Proposition $2$ in \citet{liunetzer23}. While conventional models often just assume FOSD, these authors argue that response time data can be used to test that assumption.\footnote{\citet{liunetzer23} test these inequalities in survey data and show that the null hypothesis of FOSD often cannot be rejected, in particular in cases where a probit model yields significant parameter estimates, indicating that the results of conventional models often seem to be qualitatively robust. \cite{castillo24} have subsequently used the condition to confirm the robustness of their test of salience theory, and \cite{cheng2025bullet} have used it for evaluating the causal impact of bullet-chatting on online viewer behavior.}

\subsection{Extensions}

\subsubsection{Multiple Representative Chronometric Functions}\label{subsec.mutl}

Our first extension covers the case where $\mathscr{C}^*$ is not generated from one representative profile of chronometric functions $(c_j^*)_j$ but from multiple profiles $(c_{j}^{k*})_j$ where $k \in K$. We say that $((c_{j}^{k*})_j^k,\Psi)$ generates $\mathscr{C}^*$ if for each $(c_j)_j \in \mathscr{C}^*$ there exist $(\psi_j)_j \in \Psi$ and $k \in K$ such that $(c_j)_j = (c_{j}^{k*} \circ \psi_j)_j$. Instead of a single empirical profile $(H_j)_j$, we obtain one profile $(H_{j}^k)_{j}$ for each $k \in K$, defined as in (\ref{defH}) using the respective functions $(c_{j}^{k*})_j$. Theorem \ref{t1} then becomes that property $\mathbf{P}$ is detected if $(H_{j}^k)_j$ satisfies $\mathbf{P}$ for each $k \in K$, under the otherwise identical assumptions. In words, we simply need to repeat the previous procedure for each $k \in K$ separately, and we achieve detection if we achieve it for each $k \in K$.

Among other applications, this allows us to capture additional uncertainty about the chronometric function in a convenient way. Consider an analyst who is not convinced that the chronometric effect is fully symmetric between the two options and wants to safeguard against drawing false conclusions. That analyst could work with a representative function $c^*$ that satisfies $c^{*0}(-x)=m(c^{*1}(x))$ for all $x \in \mathbb{R}_+$ and some adjustment function $m: [\underline{t},\overline{t}] \rightarrow [\underline{t},\overline{t}]$. The previous condition for detecting a positive mean, for example, then becomes
\begin{equation}\label{modtwtmean}p^0F^0(m(t)) \leq p^1F^1(t) \text{ for all } t\in[\underline{t},\overline{t}].\end{equation} 
The inequality is relaxed when $i=0$ is known to be chosen faster than $i=1$ for any given latent intensity ($m(t) \leq t$) and tightened in the opposite case ($m(t) \geq t$). When the analyst expects some asymmetry without knowing the details, the inequality can be checked for different adjustment functions $m$ that reflect various degrees and directions of asymmetry. For a robust detection of a positive mean, all modified versions (\ref{modtwtmean}) then have to hold simultaneously, ultimately requiring that the gap between LHS and RHS must be sufficiently large. This generalizes the result of \cite{AFN18} in a natural way. 

\subsubsection{Distributional Assumptions}\label{subsec.dconstr}

Our second extension covers the case where the set of admissible distributions is restricted to some $\mathscr{G}^* \subseteq \mathscr{G}$. Any such restriction reflects prior knowledge or assumptions of the analyst about properties of the distributions \citep[e.g.,][]{Manski88}. For example, in an application where each $G_j$ describes the willingness to pay of consumers minus the corresponding product price, the different distributions in the profile $(G_j)_j$ should all be horizontal shifts (by the difference in product prices) of one identical distribution. Some analysts may also be willing to impose the assumption of symmetry around the mean on all distributions in $\mathscr{G}^*$, like in conventional logit or probit models. 

Additional knowledge about the distributions can  make detection easier. A question that arises when working with restrictions on both $\mathscr{G}^*$ and $\mathscr{C}^*$ is whether the given data are \textit{rationalizable}, i.e., whether there exist distributions and chronometric functions in the restricted sets that induce the data. In Appendix \ref{App.Rat}, we study this question and provide necessary and sufficient conditions for rationalizability.\footnote{\cite{echeniquesaito17} study rationalizability for the case of deterministic responses and response times. \cite{AFN18} discuss the problem when responses and response times are stochastic. For a single binary choice problem, they show that all stochastic choice functions with response times are rationalizable when the distributions are unrestricted. For multiple binary choice problems, they discuss necessary conditions for rationalizability as well as the possibility of rationalization by generalized random utility models. They leave open the general question of rationalizability that we answer here.} 

\subsubsection{Heterogeneity and Noise}\label{subsec.stoch}

Before applying our method to data, it is natural to ask about individual heterogeneity in response speed and about noise or measurement error  \citep[see, e.g.,][]{johnson2004}. \cite{AFN18} and \cite{liunetzer23} show that several of their results have appropriate extensions that facilitate empirical testing. We follow the same approach and model the response time of an individual with realized latent value $x$ by $t = c_j(x) \cdot \eta \cdot \epsilon$, where $c_j(x)$ is the chronometric function, $\eta > 0$ captures the individual's general response speed, and $\epsilon > 0$ is a noise term comprising uncontrolled factors and measurement error. We assume that the terms $(x, \eta)$ are drawn from a joint distribution, possibly $j$-specific, with marginal cdf $G_j(x)$ whose properties we are interested in. Moreover, we assume that the analyst has access to a baseline response, such as a demographic survey question, where the response time is given by $t_b = \phi \cdot \eta$ with intensity parameter $\phi>0$. The terms $(\epsilon, \phi)$ are also drawn from a joint distribution, possibly $j$-specific, but independent of the other variables.

The baseline response allows us to normalize response times in the decision problem of interest \citep[see][]{fazio90}. Working with $\hat{t} = t/t_b = c_j(x) \cdot \epsilon/\phi$ takes care of the individual heterogeneity embodied in $\eta$. However, the analyst who uses such normalized data $(p_j,\hat{F}_j)_j$ to derive candidate distributions $(\hat{H}_j)_j$ according to (\ref{defH}) is still misspecified, as she ignores the variation embodied in $\epsilon$ and $\phi$.

In Appendix \ref{App.Stoch}, we derive conditions under which the distributions $(\hat{H}_j)_j$ still inherit a property $\mathbf{P}$ from $(G_j)_j$ and can therefore be used for testing the null hypothesis that $\mathbf{P}$ is satisfied by the true distributions. In addition to invariance properties similar to our main analysis, the key requirement is invariance to multiplicative convolutions of the form
\[\hat{G}_j(x)=\int G_j(x z) \, d \mu_j(z),\]
where $\mu_j$ is the distribution of the noise $z=\epsilon/\phi$. For single distributions, this is typically true for properties that reduce to ``pointwise arguments'' in the integral. For profiles of distributions, it may additionally require that the distributions $\mu_j$ are identical across indices $j \in J$.\footnote{There are properties for which these convolutions are not innocuous. It is known, for instance, that single-crossing of cdfs does not aggregate in general \citep{QuahStrulovici2012}. When attempting to test such a condition in noisy data, one may have to look for restricted classes of latent distributions and noise terms under which testing remains possible.} We will discuss this again in the context of hypothesis testing in Section \ref{sec.app}, where we make use of these results.

\section{Applications}\label{sec.app}

In this section, we study three applications and, in each of them, show how our theoretical result can be used to detect or reject economically meaningful distributional properties.

\subsection{Identifying an Optimal Nudge}\label{sec.app.frame}

Choices are often affected by the framing of a decision problem. Examples include the role of defaults for 401(k) savings \citep{madrian01}, the behavioral effect of phrasing the consequences of a disease as either gains or losses \citep{Tversky81}, and the relevance of the order of presentation for food choices \citep{bucher16}. The literature on nudging has promoted the idea of exploiting such framing effects to influence choices in a non-coercive manner \citep{thalersun03,thalersun08}. An important question arising in that case is how to identify the agents' true preferences and the optimal nudge \citep{benkert18}. Given that observed choices are necessarily distorted by framing, how can it be justified that a frame like a savings default helps agents make decisions that are more in line with their true preferences?

Response time data may help to answer that question. To see how, consider the following simple model of choice with frames. The true and non-distorted distribution of the latent variable in a heterogeneous population of agents is given by $G$. This variable could be a preference parameter measuring patience, a political or moral attitude, or the willingness to pay for a product. The decision problem is necessarily framed in a way that creates a bias in favor of one of the two responses. Choices under frame $f_+$ are generated by a distribution $G_+$ that is biased in favor of $i=1$, and choices under frame $f_-$ are generated by a distribution $G_-$ that is biased in favor of $i=0$. We therefore observe only the distorted fractions $p_-^0 = G_-(0)$ and $p_+^0 = G_+(0)$ but want to know the fraction $p^0 = G(0)$ of agents who truly prefer $i=0$. Knowing $p^0$ is crucial in a range of applications. For example, we need to know $p^0$ to identify which of the two frames induces the optimal choice for the majority of agents. 

We impose only one restriction on the relation between $G$ and its biased versions:
\begin{equation}\label{indifferent-agent}
    G_-(-z) = G(0) = G_+(+z) \; \text{ for some } \, z \geq 0.
\end{equation}
This condition holds, for instance, when framing induces a symmetric distortion of the true indifference point without altering its ordinal position in the distribution. In that case, an agent who is truly indifferent between the two options perceives them as having a utility difference of $-z$ or $+z$, respectively. Beyond this condition, the shapes of $G_-$ and $G_+$ can be distorted arbitrarily from $G$. The value of $z$ is not assumed to be known. For a given pair $(G_-,G_+)$ of distributions, monotonicity of the cdfs ensures that (\ref{indifferent-agent}) singles out a unique value of $p^0$ even if the condition is satisfied by multiple values of $z$. However, different pairs of distributions that are compatible with the data could exhibit different values of $p^0$. With only assumption (\ref{indifferent-agent}), it is impossible to learn the non-distorted response share when observing only the distorted choices: the observed response shares reveal $G_-(0)$ and $G_+(0)$ but not the value of these distributions at the unknown points $-z$ and $z$.

Response times provide the additional information needed to recover $p^0$. The exact value of $p^0$ is invariant to all transformations $(\psi_-,\psi_+)$  of $(G_-,G_+)$ satisfying $-\psi_{-}(-x)=\psi_{+}(x)$ for all $x\geq 0$. That is, the restrictions of $\psi_-$ to $\mathbb{R}_{-}$ and of $\psi_+$ to $\mathbb{R}_{+}$ coincide after a change of sign. A natural assumption guaranteeing that the chronometric functions are known up to these transformations is that the chronometric effect when choosing the option favored by the frame is the same in both frames.\footnote{\label{fn:default}This allows the frame to have a direct impact on response times, for example when choosing the default requires fewer actions and hence less time, as long as this does not depend on which option is the default.} Given arbitrary representative chronometric functions with that property, we can then construct distributions $(H_-,H_+)$ by (\ref{defH}) and Theorem \ref{t1} implies that we detect the value $p^0_{RT}$ for which
\begin{equation}\label{prtdetect}p^0_{RT} = p_-^0 F_-^0(t^*) = 1 - p_+^1 F_+^1(t^*)\end{equation}
holds for some $t^* \in [\underline{t},\overline{t}]$. Intuitively, a truly indifferent agent is equally fast choosing $i=0$ under frame $f_-$ as choosing $i=1$ under frame $f_+$. The mass of agents who truly prefer $i=0$ therefore equals both the mass of agents who under frame $f_-$ choose $i=0$ earlier than the indifferent agent and the mass of agents who under frame $f_+$ do not choose $i=1$ earlier than the indifferent agent. Condition (\ref{prtdetect}) determines the response time of the truly indifferent agent by equating these two masses. The term $p_-^0 F_-^0(t)$ is a continuous and weakly increasing function that starts at $0$ for $t = \underline{t}$ and ends at $p_-^0$ for $t=\overline{t}$. The term $1-p_+^1 F_+^1(t)$ is a continuous and weakly decreasing function which starts at $1$ and ends at $1-p_+^1=p_+^0 \leq p_-^0$. Hence, while the value $t^*$ that solves (\ref{prtdetect}) is not necessarily unique, the detected value $p^0_{RT}$ is unique and satisfies $p_+^0 \leq p^0_{RT} \leq p_-^0$. Response times allow us to determine where exactly the true value falls in between the two distorted values $p_+^0$ and $p_-^0$.

The literature has proposed alternative ways to determine $p^0$. For example, it is common in experimental economics and survey design to randomize features that may affect responses but are deemed irrelevant. In our context, this means that we obtain
\begin{equation}\label{pavdetect}p^0_{AV} = \frac{p_-^0 + p_+^0}{2}.\end{equation}
A more sophisticated method was proposed by \cite{goldinreck19,goldinreck20}. They assume that some agents are consistent and always choose in line with their true preferences, while other agents are inconsistent and always follow the frame. Choices against a frame then identify all consistent agents with the respective opposing preference, so that
\begin{equation}\label{pgrdetect}p^0_{GR} = \frac{p_+^0}{p_+^0 + p_-^1}\end{equation}
identifies the share among the consistent agents who truly prefer $i=0$. Under the additional assumption that consistent and inconsistent agents have the same distribution of preferences, (\ref{pgrdetect}) also identifies the value $p^0$ in the entire population.\footnote{\cite{goldinreck19,goldinreck20} suggest the assumption of identical preferences conditional on observable covariates. Since we do not use covariates to improve the other methods, we do not use them here either.}

We now proceed to compare the three predictions (\ref{prtdetect})-(\ref{pgrdetect}) empirically. This constitutes a test of the response time-based approach and all its assumptions (including symmetry) against natural benchmarks. We use the data of \citet{serragarcia23} on the impact of defaults on the willingness to vaccinate in the context of the COVID-19 pandemic. In the first wave of their experiment, \citet{serragarcia23} recruited about $600$ subjects on Prolific and elicited their intention to get vaccinated. Subjects were randomly assigned to one of three treatments: opt-out, opt-in, and active-choice. The opt-out treatment described a situation where an appointment for vaccination already exists, and the ``Leave as is and receive the vaccine'' answer option was preset as the default on the computer screen. Analogously, the opt-in treatment described a situation with no appointment as the preset default. The active-choice treatment described the problem neutrally without any preset default. Response times were recorded.\footnote{The data for all applications in Section \ref{sec.app} were collected using Qualtrics. In Subsections \ref{subsec.app.happy} and \ref{sec.predicting}, we follow \cite{liunetzer23} and define response time by the time of the last click before submission of the answer. We have to deviate from this rule here in Subsection \ref{sec.app.frame}, because the existence of a preset default answer made it possible to submit that answer without any earlier clicks, so that the time of the last click is often missing. We therefore use the time of submission of the answer here. Recall that it is not a problem for our approach if choosing the default is systematically faster because it requires fewer clicks than choosing against the default, as discussed in footnote \ref{fn:default}.} While a neutral framing is often difficult or impossible to achieve in real-world applications, the existence of the neutral active-choice treatment in the experiment allows us to evaluate the predictions of all three methods against a ground truth.

The predictions are, despite being in principle designed for a within-subjects application, valid for the experimental between-subjects design under the assumption that random allocation of subjects to treatments generates samples representative of the population. For our response time-based prediction, we normalize response times by subtracting in logs a baseline response time, to account for individual heterogeneity (see Section \ref{subsec.stoch}). We use as a baseline the total duration it took subjects to complete the study net of the response time in the question about vaccination.

In Figure \ref{fig:share_unvaccinated}, the blue bars show the shares of subjects who decide to remain unvaccinated under opt-in ($p_-^0=0.342$), opt-out ($p_+^0=0.303$), and active-choice ($p^0=0.335$). The orange bars are our predictions for the active-choice outcome based on response time ($p_{RT}^0=0.338$), simple averaging ($p_{AV}^0=0.323$), and Goldin-Reck ($p_{GR}^0=0.315$). The response time-based method clearly outperforms both alternatives and gets very close to the actual choice share in the active-choice treatment.

\begin{figure}[t]
  \centering
  \includegraphics[scale=0.4]{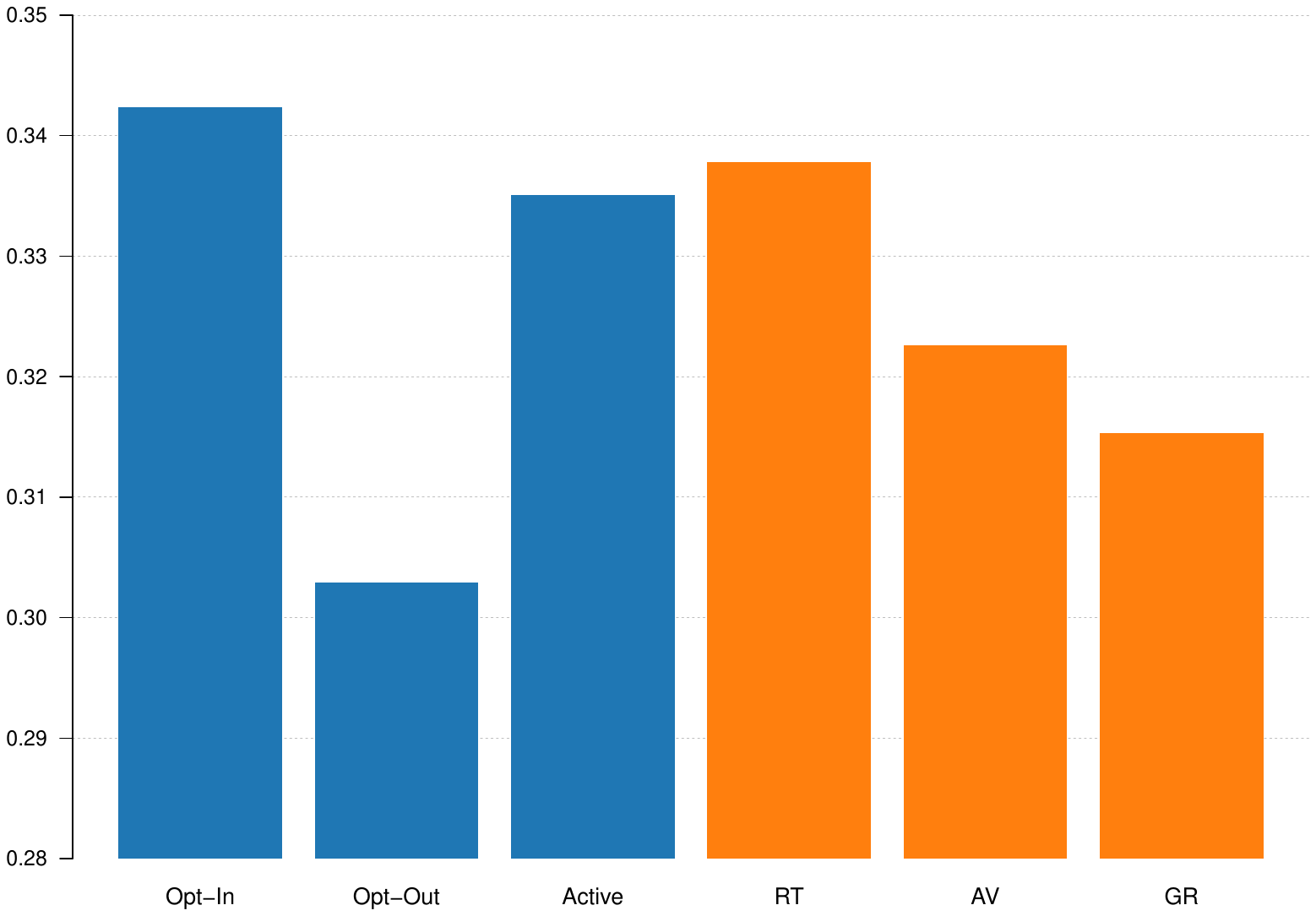}
  \vspace{-2mm}
  \caption{Shares unvaccinated}\label{fig:share_unvaccinated}
  \vspace{3mm}
\begin{minipage}{\textwidth}
\small\textit{Notes:} The blue bars depict the shares of subjects who decide to remain unvaccinated under the treatments opt-in, opt-out, and active-choice in the experiment of \citet{serragarcia23}. The orange bars depict the predictions of the share of unvaccinated subjects under active-choice for the three methods based on response times (RT), simple averaging (AV), and Goldin-Reck (GR).
\end{minipage}
\end{figure}

We can also ask which of the two defaults induces the optimal choice for a larger number of subjects. Among the subjects who respond to the frame, a fraction $(p^0 - p_+^0)/(p_-^0 - p_+^0) = 0.816$ truly prefer to remain unvaccinated, and hence opt-in is optimal for a larger number of responsive subjects.\footnote{That statement about the optimal frame is subject to various qualifications. For example, it ignores possible welfare effects of the default on subjects who do not respond to changes in the default. It also ignores other welfare effects of vaccinations due to externalities etc.} If we replace the true value $p^0$ in this expression by the estimated value $p_{RT}^0$, then we arrive at the same conclusion and thus correctly identify the optimal nudge. By contrast, using $p_{AV}^0$ by definition always implies that both frames are equally good, while using $p_{GR}^0$ here suggests opt-out as optimal.

To get an impression of how variable the three predictions are, we sample with replacement from the opt-in and opt-out treatment groups to generate synthetic treatment groups of the same size as the original data. We do this 5'000 times and each time compute the predictions and their difference to the active-choice benchmark, which we keep fixed.\footnote{If the experiment were a within-subjects design, it would be appropriate to sample subjects and then use their choices across all three treatments in the synthetic samples. This is not feasible with the current between-subjects design, and hence we keep the active-choice benchmark as the fixed ground truth. When sampling from the original data, it can happen that the fraction of unvaccinated subjects becomes smaller under opt-in than under opt-out, effectively corresponding to a directional reversal of the standard default effect. We let the data speak in those cases and treat opt-in as the frame favoring vaccination.} The response time-based prediction remains the most precise estimate in 56.7\% of the samples, followed by averaging in 20.8\% and Goldin-Reck in 22.5\% of all samples. In 2,743 of the 5,000 resamples, the active-choice share lies between the two resampled frame shares, so the fraction of switchers favoring either choice is defined. On these draws, the response time-based method identifies the optimal nudge in 65.8\% of cases and Goldin-Reck in 33.8\%. Averaging predicts a tie and has an expected score of 50.0\% under fair random tie-breaking.\footnote{For comparison, a nearest-frame directional score that includes all 5,000 draws is 70.1\% for the response time-based method, 50.0\% for averaging, and 28.2\% for Goldin-Reck. The nearest-frame score is a distance comparison when the active-choice share lies outside the frame interval.}

Overall, this application shows that response times can recover meaningful preference information even when choices are distorted by a frame. Our method yields a reliable estimate of the non-distorted choice share and identifies which frame best aligns with underlying preferences. The empirical comparison highlights that this approach can outperform existing alternatives, both in accuracy and in its ability to guide the choice of an optimal nudge.

\subsection{Testing Decreasing Marginal Happiness}\label{subsec.app.happy}

\citet{Easterlin2005} postulates that ``[f]ew generalizations in the social sciences enjoy such wide-ranging support as that of diminishing marginal utility of income'', where he interprets utility explicitly as ``subjective well-being'' (p.\ 243). Yet the evidence is surprisingly fragile. Easterlin himself showed that the principle does not extend from cross-section to time series data. Later studies, including \citet{Oswald08} and \citet{KaiserOswald22}, emphasized that even cross-sectional patterns are inconclusive. The difficulty is that the observed relationship between income and reported happiness is the composition of two functions: the true mapping $h: W \rightarrow H$ from income to happiness and the reporting function $r: H \rightarrow R$ from happiness to survey response. Observing that $r(h(w))$ is concave does not imply that $h(w)$ is concave and therefore does not establish decreasing marginal happiness. The relationship may just as well be due to a concave reporting function $r(h)$. Other studies have estimated advanced ordered response models, but as \cite{BL19} have shown, their conclusions depend on the distributional assumptions of the model. \cite{KaiserOswald22} conclude that the problem is ``fundamental, little recognized, and so far unsolved'' (p.\ 3).

To formalize the income-happiness relation in our framework, let $(G_w)_{w}$ denote the family of happiness distributions for all possible income levels $w \in W \subseteq \mathbb{R}_+$ and define $\mu(w)= \int_{\mathbb{R}} x dG_w(x)$ as the average happiness of agents with income $w$. The question of interest is whether $\mu$ is a concave function of $w$. A natural starting point is to plot the average response in a happiness survey against income, as many studies have done. We repeat this exercise here, using MTurk survey data of \cite{liunetzer23} in which 3'743 respondents answered a binary question on whether they are ``rather happy'' or ``rather unhappy.'' Household income in the survey was reported in three broad bins: below \$40'000, between \$40'000 and \$69'999, and above \$70'000. To proceed, we need to associate each of these bins with a unique income value, denoted $w_L$, $w_M$ and $w_H$. As a first illustrative benchmark, we take the midpoint of each bin, assuming an upper bound of \$110'000 for the high-income bin. Figure \ref{fig:conc} shows that the share of ``rather happy'' respondents rises with income but at a diminishing rate, producing a concave pattern. Like in \citet{Easterlin2005} and the literature cited therein, this is a bivariate relation without any additional controls, but the results are similar when conditioning on various socio-demographic variables.

\begin{figure}[t]
\centering
     \includegraphics[scale=0.5]{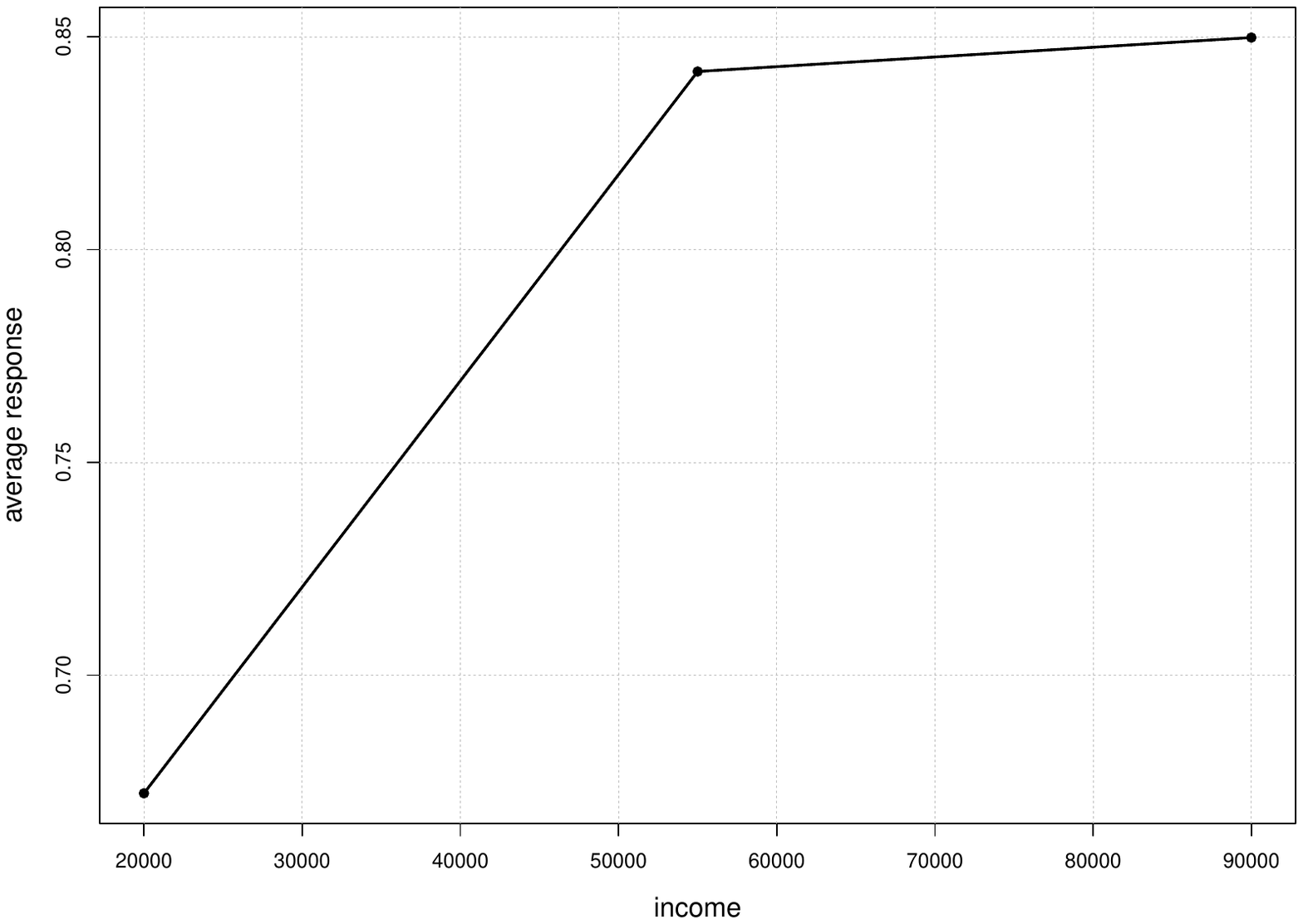}
       \caption{The relationship between income and reported happiness.}\label{fig:conc}
        \vspace{3mm}
   \begin{minipage}{\textwidth}
   \small\textit{Notes:} The curve corresponds to the case where the income for each bin is given by the midpoint of the bin, assuming an upper bound of \$110'000 for the high-income bin. Responses are coded as $0$ (``rather unhappy'') and $1$ (``rather happy''), so the average coincides with the share of respondents reporting ``rather happy.''
  \end{minipage}
\end{figure}

Our point is that concavity in Figure \ref{fig:conc} does not establish concavity of $\mu$, which here reduces to the single inequality 
\begin{align}\label{eq.midpoint_concavity}
\alpha \mu(w_L) + (1-\alpha) \mu(w_H) \leq \mu(w_M),
\end{align}
where $\alpha$ is such that $\alpha w_L + (1-\alpha) w_H = w_M$ (and which becomes $\alpha=0.5$ in our illustrative benchmark). Following the logic of \cite{BL19}, we can explain Figure \ref{fig:conc} with happiness distributions $(G_L,G_M,G_H)$ for which (\ref{eq.midpoint_concavity}) is violated, for example distributions that become more right-skewed as income grows and which therefore go along with high average happiness among the rich. It is impossible to reject or verify such distributional assumptions based on response data alone.

According to (\ref{eq.midpoint_concavity}), decreasing marginal happiness is equivalent to a ranking of the means of two distributions. The LHS of (\ref{eq.midpoint_concavity}) represents the mean happiness of a population that mixes low- and high-income individuals in proportions $\alpha$ and $1-\alpha$, while the RHS is the mean happiness of the middle-income group. We now develop a response time-based criterion that allows us to rank these means. 

The ranking of the means of two distributions $G_A$ and $G_B$ is invariant to positive affine transformations that are identical for the two groups. Detecting it requires knowledge of the chronometric functions up to common and positive linear transformations. Since this may be too demanding, we can work with a sufficient condition. Consider the relative asymmetry
\[G_A(x)+G_A(-x) \leq G_B(x)+G_B(-x) \; \text{for all} \; x \in \mathbb{R}_+,\]
which implies that the mean of $G_A$ is larger than the mean of $G_B$ (see Appendix \ref{App.Rank}). This condition is invariant to transformations that are identical across groups and symmetric around zero. Therefore, an assumption guaranteeing that it can be detected from response time data is that the chronometric effect is the same for the two options and the two groups, but otherwise unrestricted. Theorem \ref{t1} now implies that the condition is detected if it holds for functions $H_A$ and $H_B$ constructed based on (\ref{defH}). This can once more be expressed directly in terms of the observed data, and it follows that
\begin{align}\label{avcond}
p_A^0 F_A^0(t) - p_B^0 F_B^0(t) \leq p_A^1 F_A^1(t) - p_B^1 F_B^1(t) \; \text{for all} \; t \in [\underline{t},\overline{t}]
\end{align}
is sufficient for detecting that the mean of $G_A$ is larger than the mean of $G_B$.\footnote{Condition (\ref{avcond}) was first derived in our earlier working paper \citet{liunetzer20}. We referenced it in a footnote of \cite{liunetzer23} but have not published it otherwise.} Intuitively, the requirement that the latent distribution of group $A$ is relatively more asymmetric towards large values than the distribution of group $B$ translates into the empirical condition that response differences between group $A$ and group $B$ must be larger for $i=1$ than for $i=0$, even when considering only those responses that happened at any $t$ or earlier.

We use this condition to test the hypothesis that marginal happiness is decreasing in income, as formalized in (\ref{eq.midpoint_concavity}).\footnote{There are similar yet different approaches in the literature. Already in his 1996 opus magnum \textit{Infinite Jest}, novelist David Foster Wallace described a related idea. The novel centers around a lost movie that is so entertaining that people who have seen it are willing to sacrifice everything, including their toes, to see it again. A terrorist organization wants to use the movie as a weapon and contemplates how to test that the marginal entertainment of the movie is not decreasing. They propose checking that the time it takes for people to accept sacrificing another toe is not increasing after repeated views of the movie. We remark that this procedure does not cleanly distinguish between the marginal entertainment of the movie and the marginal cost of losing a toe. A non-fiction investigation of how response times vary with value levels in an individual decision-making task was done by \citet{shevlin22}. They show that response times tend to decrease in the value level, incompatible with decreasing marginal utility. \citet{liunetzer23} have shown that, using an ordered probit model, happiness increases significantly when moving from the low to the middle income category and that, using response times, the hypothesis of FOSD of the happiness distributions cannot be rejected. The effect of moving from middle to high income is insignificant and FOSD can be rejected. These results do not yet establish concavity of the income-happiness relation.} Under the assumption that all three income groups have the same chronometric function, which our criterion requires anyway, we mix the two extreme income groups with appropriate weights to obtain the cdf $G_P(x)= \alpha G_L(x) + (1-\alpha) G_H(x)$ of happiness in the pooled group. Condition (\ref{avcond}) then becomes the sufficient condition to detect decreasing marginal happiness
\begin{align}\label{null_avcond}
p_M^0F_M^0(t) - p_P^0F_P^0(t) \leq p_M^1F_M^1(t) - p_P^1F_P^1(t) \; \text{for all} \; t \in [\underline{t},\overline{t}],
\end{align}
where $p_P^i F_P^i(t)= \alpha p_L^i F_L^i(t) + (1-\alpha) p_H^i F_H^i(t)$. Intuitively, this inequality rules out examples like the increasingly right-skewed distributions discussed above, as these would generate relatively fast happy and slow unhappy responses in the pooled group.

We follow \cite{liunetzer23} and normalize response times by subtracting in logs the response time to a baseline marital-status question. This normalization corroborates further the assumption of identical chronometric functions in the different income groups.\footnote{Our results in Subsection \ref{subsec.stoch} imply that all properties which we investigate here can be tested using the response time-based detection conditions: under the null hypothesis that the property holds, its detection condition must be satisfied in the normalized data even when there is noise or measurement error that is identically distributed across the income groups being compared.} Figure \ref{fig:fosd} plots the two sides of (\ref{null_avcond}) for the illustrative benchmark of $\alpha=0.5$. 

\begin{figure}[t]
  \centering
  \includegraphics[scale=0.5]{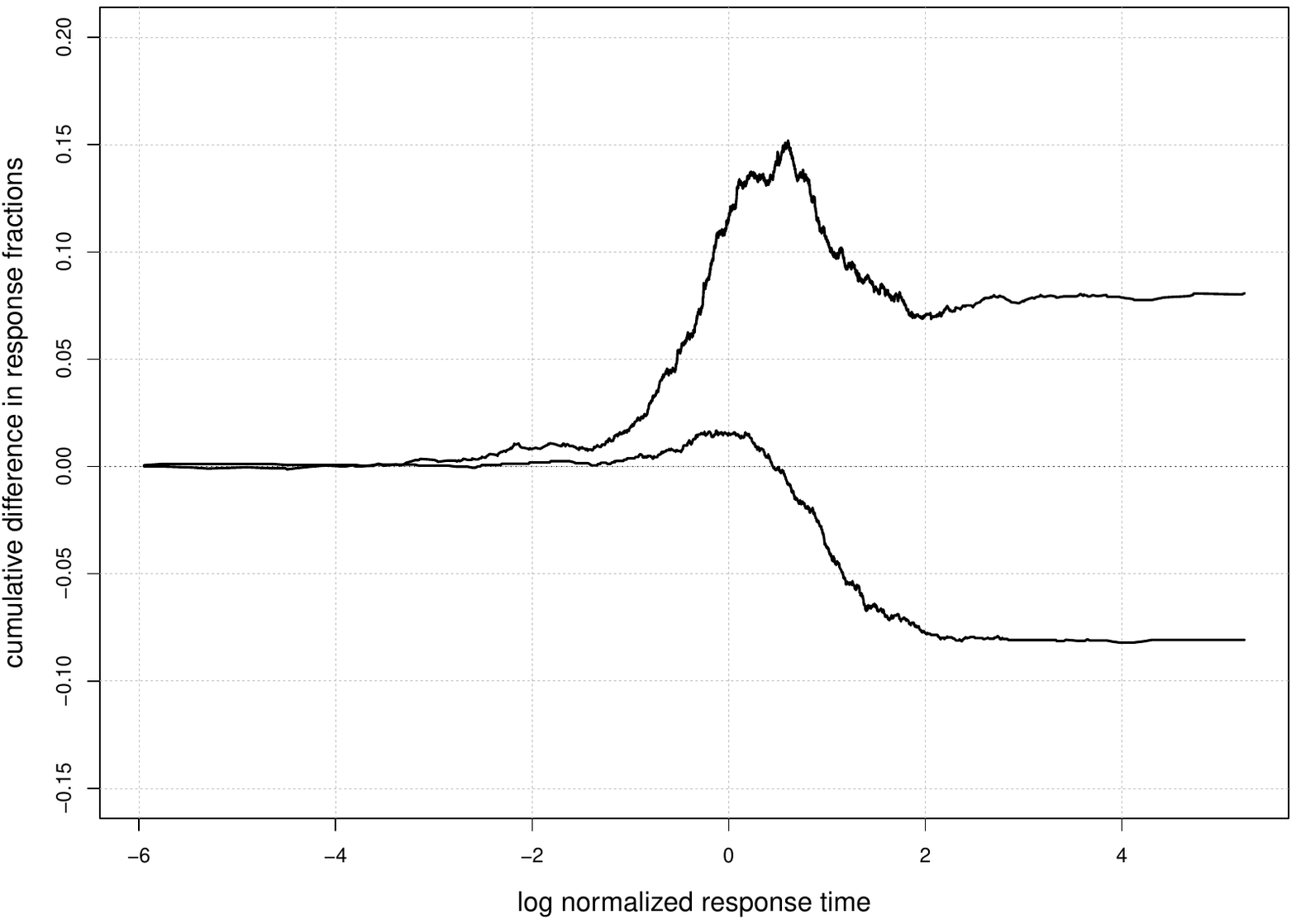}
  \caption{Empirical conditions for testing the income-happiness relation.}\label{fig:fosd}
  \vspace{3mm}
\begin{minipage}{\textwidth}
\small\textit{Notes:} The curve at the top represents the empirical function $p_M^1 F_M^1(t)- p_P^1 F_P^1(t)$, while the one at the bottom represents $p_M^0 F_M^0(t) - p_P^0 F_P^0(t)$, both depicted for the case $\alpha=0.5$.
\end{minipage}
\end{figure}

The inequality is not exactly satisfied in the data, because the LHS exceeds the RHS for some small $t$. To assess statistical significance, we employ a bootstrap-based method as in \cite{liunetzer23} that rests on a test by \citet{barrett2003consistent}. To ensure robustness, we also go beyond the benchmark $\alpha=0.5$ and conduct the test separately for each value $\alpha \in \{0.1, 0.2, \ldots, 0.9\}$, i.e., we vary the mixing weight used to construct the pooled group $P$.\footnote{Different choices of representative incomes for the survey bins imply different values of $\alpha$ through the identity $\alpha w_L + (1-\alpha) w_H = w_M$. For instance, changing the upper bound for the high-income group from \$110'000 to \$150'000 or \$200'000 results in $\alpha=0.611$ or $\alpha=0.696$, respectively. Alternatively, one could use more granular income data about MTurk subjects drawn from other studies to approximate the average income in each bin. Doing this with the data from \citet[Table 2]{moss2023ethical}, and using an upper bound of \$200'000 for the high-income bin, we obtain $\alpha=0.646$.} As Table \ref{tab:concavity_results} shows, the p-values for the null hypothesis that condition (\ref{null_avcond}) holds are consistently large. We clearly cannot reject our sufficient condition for decreasing marginal happiness.

\begin{table}[t]
    \centering
    \begin{tabular}{ccc}
    \hline
    $\alpha$ & \multicolumn{2}{c}{p-value} \\
     & (\ref{null_avcond}) & (\ref{null_htcond}) \\
    \hline
    0.1 & 0.962 & 0.188 \\
    0.2 & 0.961 & 0.195 \\
    0.3 & 0.958 & 0.232 \\
    0.4 & 0.956 & 0.282 \\
    0.5 & 0.955 & 0.337 \\
    0.6 & 0.950 & 0.384 \\
    0.7 & 0.942 & 0.444 \\
    0.8 & 0.941 & 0.503 \\
    0.9 & 0.934 & 0.558 \\
    \hline
    \end{tabular}
    \caption{Summary of tests for concavity.}
    \label{tab:concavity_results}
\end{table}

A stronger sufficient condition for the ranking of two means is first-order stochastic dominance, discussed in Section \ref{subsectillu}. We therefore also consider the property that $G_M$ first-order stochastically dominates $G_P$, which is a fortiori sufficient for decreasing marginal happiness, but whose detection does not require symmetry of the chronometric functions since no comparisons between positive and negative values of the latent variable are involved. The testing condition here becomes
\begin{align}\label{null_htcond}
p_M^0 F_M^0(t) - p_P^0 F_P^0(t) \leq 0 \leq p_M^1 F_M^1(t) - p_P^1 F_P^1(t) \; \text{for all} \; t \in [\underline{t},\overline{t}],
\end{align}
which in Figure \ref{fig:fosd} requires that the functions are separated by zero. While not exactly satisfied in the data, Table \ref{tab:concavity_results} shows that the p-values remain large.\footnote{The test implements the procedure based on \citet{barrett2003consistent} together with a joint hypothesis correction by \cite{romanowolf2016} to account for the fact that condition (\ref{null_htcond}) contains two inequalities. We refer the reader to \citet{liunetzer23} for details.} We clearly cannot reject even this stronger sufficient condition for decreasing marginal happiness.

Finally, we investigate the hypothesis that marginal happiness increases with income, which corresponds to reversing all inequalities studied so far. The reversed versions of \eqref{null_avcond} and \eqref{null_htcond} are rejected at all plausible levels of statistical significance (all $p<0.001$).

To summarize, our results are consistent with the idea that marginal happiness is decreasing in income in a cross-sectional data set. While our tests avoid several of the pitfalls noted earlier in this section, some limitations remain. Most notably,  we were only able to test the concavity of expected happiness for three income levels implied by the survey bins, rather than across the entire income distribution. Additionally, the income-happiness relation is bivariate without controlling for potential confounding factors. Future research that addresses these limitations offers significant potential.

\subsection{Predicting Treatment Heterogeneity}\label{sec.predicting}

Policymakers and experimenters frequently need to decide which populations to target with a treatment. For instance, a policymaker may want to make costly employment training available only to the most responsive subgroup of the unemployed. An online platform may want to increase prices for those consumers who are least likely to reduce demand. A tax authority following the \cite{ramsey27} rule wants to impose larger taxes on markets with smaller elasticities. An experimenter may want to impose some treatment on the least and the most responsive group to obtain bounds on the treatment effect.

Modern methods for heterogeneous treatment effects and policy learning, such as causal forests or empirical welfare learning, rely on post-treatment outcomes to estimate subgroup effects \citep[e.g.,][]{Athey2016PNAS,Kitagawa2018ECTA}. Likewise, behavioral work on treatment moderators relies on comparisons at the post-treatment level \citep{Krefeld2024PNAS}. In other words, the treatment must be imposed on a sample of the entire population before the eventual target subgroups can be identified. We show that the ordering of treatment effects across subgroups can be predicted pre-treatment, using choices and response times from a baseline before any treatment even took place. As such, response times may be used as a pre-outcome data source when developing statistical treatment rules under limited information as in \citet{Manski2004Statistical}.

Let $G_A$ and $G_B$ denote the latent distributions for groups $A$ and $B$, respectively. Before treatment, the probability that an agent in group $j$ chooses $i=1$ is $1 -  G_j(0)$. Suppose that the agents then receive a treatment that favors option $i=1$. We assume that the treatment operates by lowering the decision threshold from \(0\) to some \(x^*\leq 0\). That is, under treatment, an agent chooses \(i=1\) if and only if \(x>x^*\). For example, if the treatment raises the payoff of option \(i=1\) by \(\delta\geq 0\) for all agents, then the induced cutoff is \(x^*=-\delta\). More generally,  our formulation allows for nonlinear treatment effects that vary with the latent value \(x\), provided that their choice implications can be summarized by a shift in the decision threshold. It follows that, under treatment, an agent in group \(j\) chooses \(i=1\) with  probability \(1-G_j(x^*)\). The treatment effect for group \(j\) is therefore $G_j(0) - G_j(x^*) \geq 0$, and it is larger for group $A$ than for group $B$ whenever
\begin{equation}\label{ranktreatcond}
G_B(0) - G_B(x^*) \leq G_A(0) - G_A(x^*).
\end{equation}
If we do not know the exact cutoff generated by the treatment, we can still predict that group \(A\) will react more strongly than group \(B\) when we detect that \eqref{ranktreatcond} holds for all \(x^*\leq 0\). This property is invariant to transformations $(\psi_A,\psi_B)$ which satisfy $\psi_A(x)=\psi_B(x)$ for all $x \leq 0$ and $\psi_A(0)=\psi_B(0)=0$. A natural assumption guaranteeing that the chronometric functions are known up to these transformations is that the chronometric effect when choosing $i=0$ is the same for the two groups. Then, under this assumption, the distributions $H_j$ constructed by (\ref{defH}) satisfy the desired inequality if and only if
\begin{equation}\label{ranktreatdetect}
p_B^0 (1 - F_B^0(t)) \leq p_A^0 (1 - F_A^0(t)) \text{ for all } t \in [\underline{t},\overline{t}].
\end{equation}
This condition rests on data from choices of $i=0$ only, as these are the ones that the treatment may potentially shift. The condition requires $p_B^0 \leq p_A^0$ and, additionally, that the choices of group $A$ must be sufficiently slow compared to group $B$. Intuitively, considering the agents who take at least time $t$ to respond picks out those who are relatively close to indifferent, at the resolution defined by $t$. If group $A$ has more of this near-indifferent mass than group $B$ for any $t$, then a treatment that shifts the cutoff to the left converts more mass into $i=1$ choices in group $A$ than in group $B$.\footnote{The literature has documented several times that agents who respond more slowly are more likely to change their behavior over time \citep[e.g.][]{fazio86,bassili93,konovalov2019revealed}, across elicitation modes \citep[e.g.][]{AGKW16}, or in response to counterarguments \citep[e.g.][]{bassili91,huckfeldt99}. However, there are also experiments in which response time does not systematically predict behavioral change \citep{bassili00}. Our contribution here is to derive a precise criterion for predicting treatment heterogeneity and to test this criterion empirically.} Taken together, (\ref{ranktreatdetect}) is a sufficient condition that allows us to rank groups by responsiveness using choices and response times pre-treatment and no further distributional assumptions.

If an analyst has additional information that bounds the strength of the treatment and thereby rules out sufficiently low cutoffs \(x^*\), the detection condition can be weakened. For example, the analyst may have prior knowledge that the treatment is not strong enough to shift the decisions of the \(q\%\) most committed supporters of option \(i=0\). Inequality (\ref{ranktreatdetect}) then needs to be satisfied only for all $t \in [t_q,\overline{t}]$, where $t_q$ is the $q/100$ quantile of response times conditional on choosing $i=0$ in the pooled groups, defined by $F^0_{\mathrm{pool}}(t_q)=q/100$. Thus, the bound on the strength of the treatment translates into a natural weaker condition: we only need to compare the near-indifferent masses at sufficiently slow response times. Consequently, and in contrast to the unbounded case, the detection condition does not necessarily require $p_B^0 \leq p_A^0$, that is, we could have a larger number of agents favoring $i=0$ in group $B$ than in group $A$ and still predict a stronger treatment response from group $A$.

We put these results to work using data from \citet{Krefeld2024PNAS}. In their Study 3, a total of 1'200 subjects recruited on Prolific and MTurk responded to the Trolley Problem. In this moral dilemma, a runaway trolley threatens the lives of five people. In the \textit{bridge case}, which \citet{Krefeld2024PNAS} label as the baseline condition, the only way to stop the trolley is by pushing a man onto the tracks. In the \textit{track case}, which they label as the treatment condition, pulling a lever diverts the trolley onto a side track, killing one instead. While the two scenarios are identical in that an intervention results in only one death instead of five, they are typically judged differently by subjects. In the data of \citet{Krefeld2024PNAS}, only 43.3\% of all subjects report that pushing the man onto the tracks in the bridge case is morally admissible, while 83.6\% of all subjects report that pulling the lever in the track case is morally admissible. \citet{Krefeld2024PNAS} are interested in treatment heterogeneity at the subject level. They develop a factor model encompassing the effect moderators Fluid Intelligence (F), Attentiveness (A), Crystallized Intelligence (C), and Experience (E). They find that attention is the only moderator that has a statistically significant influence on the treatment effect, with more attentive subjects reacting more strongly.

Our goal is to use subjects' responses and response times in the bridge case to predict how the treatment effect of moving to the track case differs across subgroups. In terms of our above model, option $i=1$ corresponds to the response that the intervention is morally admissible, which goes up with the treatment. We construct pairs of groups using the demographic information available in the study as well as the four factors that \citet{Krefeld2024PNAS} consider as moderators. To predict which of two groups has the stronger treatment effect, we test condition \eqref{ranktreatdetect} in both directions. If only one of these hypotheses is rejected, we predict that the treatment effect is larger for the respective group. If both or none of the hypotheses are rejected, we do not make a prediction.\footnote{The test uses bootstrapping as described in Section \ref{subsec.app.happy}. Our results in Subsection \ref{subsec.stoch} ensure the validity of our tests even when there is noise or measurement error that is identically distributed across the two subgroups being compared. We again log-normalize response times. We proceed as in Section \ref{sec.app.frame} and use the total duration it took subjects to complete the study net of the focal question's response time.} We then check the correctness of our prediction using the actual data from the track case.

To form subgroups, we compare women and men for gender, and separated groups on either side of the median for the other dimensions. Median ties and central ordinal categories are excluded, and each comparison uses the observations available for that dimension. Table \ref{tab:treatment_pred_summary} reports p-values, predictions and actual outcomes for all pairs of groups. In 8 out of 9 cases, we are able to make a prediction at the 5\% significance level. All of these predictions turn out to be correct, which means that the prediction agrees with the observed ordering of the treatment effects.\footnote{If one instead reverses the labeling of baseline and treatment, using track-case choices and response times for the prediction and the bridge case as the comparison condition, the method yields 7 rather than 8 predictions at the 5\% level. All of them are correct. There is no prediction for education.} 

In particular, it is instructive to compare our predictions for the FACE moderators to \citet{Krefeld2024PNAS}. They find that more attentive subjects exhibit a significantly stronger treatment effect, while not documenting any significant effects for the other three moderators. Our analysis, based on pre-treatment data only, reveals the same prediction for attention. Moreover, our (significant) predictions for the other three moderators broadly agree with the (insignificant) signs reported in \citet[Table S17]{Krefeld2024PNAS}.\footnote{More precisely, in their specification with only the four moderators, the signs agree for all four moderators. In their specification with the moderators and a dummy for the panel (Prolific/MTurk) the signs agree for attention, crystallized intelligence, and fluid intelligence, but not for experience.}

\begin{table}[t]
\centering
\begin{tabular}{@{} l r r r r @{}}
\toprule
Dimension (group $A$ vs. $B$) & $p(A>B)$ & $p(B>A)$ & Prediction & Correct \\
\midrule
Gender (women vs men)                 & 0.866 & $<0.001$ & $A$ stronger & True \\
Politics (liberal vs conservative)    & 0.977 & $<0.001$ & $A$ stronger & True \\
Income (poor vs rich)                 & 0.191 & 0.304 &     -       &   - \\
Education (low vs high)               & 0.809 & $<0.001$ & $A$ stronger& True \\
Age (young vs old)                    & $<0.001$ & 0.410 & $B$ stronger & True \\
Face F (top vs bottom)                & 1.000 & $<0.001$ & $A$ stronger & True \\
Face A (top vs bottom)                & 0.951 & $<0.001$ & $A$ stronger & True \\
Face C (top vs bottom)                & $<0.001$ & 0.836 & $B$ stronger & True \\
Face E (top vs bottom)                & $<0.001$ & 0.892 & $B$ stronger & True \\
\bottomrule
\end{tabular}
\caption{Testing the ranking condition in \eqref{ranktreatdetect} in both directions.}
\label{tab:treatment_pred_summary}
\vspace{3mm}
\begin{minipage}{\textwidth}
\small\textit{Notes:} The table reports p-values from testing the sufficient pre-treatment ranking condition in \eqref{ranktreatdetect} in both directions for each pair of subgroups (group $A$ vs.\ $B$). We predict which subgroup exhibits the stronger treatment effect when exactly one direction is rejected at the 5\% level; otherwise no prediction is made.
\end{minipage}
\end{table}

Our approach does not yield a prediction for the subgroup split based on income, because neither direction of inequality \eqref{ranktreatdetect} can be rejected. A possible way forward is to assume that the treatment effect is bounded and to exclude the most committed (fastest) subjects from the test. However, even in that case we do not obtain a prediction for income at the 5\% significance level.


Overall, our analysis shows that response times can be used to anticipate heterogeneous treatment effects before the intervention even takes place. By using our detection condition \eqref{ranktreatdetect} as a testable inequality, we can successfully rank groups by anticipated responsiveness using only baseline choices and response times. As in our earlier applications, this is both a test of our theory and a proof of the value of using response time data in practical applications.

\section{Conclusion}\label{sec.concl}

The goal of this paper is to provide a systematic account of how response time data can be used in binary response models. Our main theoretical result relates the set of identifiable distributional properties to the set of admissible chronometric functions. The fundamental idea is that the joint distribution of responses and response times identifies a composition of the latent distribution and the chronometric function. Properties of the distribution that are preserved under a given set of transformations can therefore be identified if the chronometric function is known up to these transformations. Several existing results in the literature follow as corollaries and can be generalized, and many new results emerge. We demonstrate the usefulness of our approach in three empirical applications: identifying an optimal nudge, testing decreasing marginal happiness of income, and predicting treatment heterogeneity. In each case, incorporating response times yields new economic insights that cannot be obtained from conventional methods alone.

Our applications in Section \ref{sec.app} are merely examples of the scope of the method and not an exhaustive list. Appendix \ref{sec.theory.app} presents additional distributional properties that can be studied within our framework, including the ranking of inequality or dispersion, likelihood-ratio dominance, unimodality, and correlations between latent variables and observable characteristics. These properties arise naturally in many empirical settings. Applications that we have in mind include the study of polarization using surveys on political attitudes, evaluating the magnitude of under-reporting on sensitive views based on behavioral records in the lab, and optimal product pricing using data on purchase decisions from online platforms. There are also several possible extensions of our framework that merit further investigation. These include correlations among multiple latent variables and the case with more than two choice options.

\bibliographystyle{aea}
\bibliography{Cardinal}

@article{AGKW16,
  author= {Al\'{o}s-Ferrer, Carlos and Grani\'{c}, Dura-Georg and Kern, Johannes and Wagner, Alexander K.},
  title= {Preference Reversals: Time and Again},
  year= {2016},
  journal={Journal of Risk and Uncertainty},
  volume={52},
  pages={65--97}
}

@unpublished{AFFG24,
  author= {Carlos Al\'{o}s-Ferrer and Ernst Fehr and Helga Fehr-Duda and Michele Garagnani},
  title= {Common Ratio and Common Consequence Effects Arise From True Preferences},
  year= {2024},
  note={Mimeo}
}

@article{AFN18,
  author= {Carlos Al\'{o}s-Ferrer and Ernst Fehr and Nick Netzer},
  title= {Time Will Tell: Recovering Preferences When Choices Are Noisy},
  year= {2021},
  journal={Journal of Political Economy},
  volume={129},
  number={6},
  pages={1828--1877}
}

@article{AG22,
  author= {Carlos Al\'{o}s-Ferrer and Michele Garagnani},
  title= {Strength of Preference and Decisions Under Risk},
  year= {2022},
  journal={Journal of Risk and Uncertainty},
  volume={64},
  pages={309--329}
}

@unpublished{AFG23,
  title={Identifying Nontransitive Preferences},
  author={Al\'{o}s-Ferrer, Carlos and Garagnani, Michele and Fehr, Ernst},
  year={2024},
  note={Mimeo}
}

@article{arrow59,
  title={Rational Choice Functions and Orderings},
  author={Arrow, Kenneth J.},
  journal={Economica},
  volume={26},
  year={1958},
  number={102},
  pages={121--127}
}

@article{athey2002monotone,
  title={Monotone Comparative Statics Under Uncertainty},
  author={Athey, Susan},
  journal={The Quarterly Journal of Economics},
  volume={117},
  number={1},
  pages={187--223},
  year={2002},
  publisher={MIT Press}
}

@article{Athey2016PNAS,
  author  = {Athey, Susan and Imbens, Guido W.},
  title   = {Recursive Partitioning for Heterogeneous Causal Effects},
  journal = {Proceedings of the National Academy of Sciences},
  year    = {2016},
  volume  = {113},
  number  = {27},
  pages   = {7353--7360},
  doi     = {10.1073/pnas.1510489113},
  url     = {https://www.pnas.org/doi/10.1073/pnas.1510489113}
}

@article{atkinson1970measurement,
  title={On the Measurement of Inequality},
  author={Atkinson, Anthony B},
  journal={Journal of Economic Theory},
  volume={2},
  number={3},
  pages={244--263},
  year={1970}
}

@article{barrett2003consistent,
  title={Consistent Tests for Stochastic Dominance},
  author={Barrett, Garry F. and Donald, Stephen G.},
  journal={Econometrica},
  volume={71},
  number={1},
  pages={71--104},
  year={2003},
  publisher={Wiley Online Library}
}

@article{bassili91,
  title={Response-Time Measurement in Survey Research: A Method for {CATI} and a New Look at Nonattitudes},
  author={Bassili, John N. and Fletcher, Joseph F.},
  journal={Public Opinion Quarterly},
  volume={55},
  number={3},
  pages={331--346},
  year={1991}
}

@article{bassili93,
  title={Response Latency Versus Certainty as Indexes of the Strength of Voting Intentions in a CATI Survey},
  author={Bassili, John N.},
  journal={Public Opinion Quarterly},
  volume={57},
  number={1},
  pages={54--61},
  year={1993}
}

@article{bassili00,
  title={Do Strength-Related Attitude Properties Determine Susceptibility to Response Effects? New Evidence From Response Latency, Attitude Extremity, and Aggregate Indices},
  author={Bassili, John N. and Krosnick, Jon A.},
  journal={Political Psychology},
  volume={21},
  number={1},
  pages={107--132},
  year={2000}
}

@article{benkert18,
  title={Informational Requirements of Nudging},
  author={Benkert, Jean-Michel and Netzer, Nick},
  journal={Journal of Political Economy},
  volume={126},
  number={6},
  pages={2323--2355},
  year={2018},
}

@article{BL19,
  title={The Sad Truth About Happiness Scales},
  author={Bond, Timothy N. and Lang, Kevin},
  journal={Journal of Political Economy},
  volume={127},
  number={4},
  pages={1629--1640},
  year={2019},
  publisher={The University of Chicago Press Chicago, IL}
}

@ARTICLE{bucher16,
  author = {Bucher, T. and Collins, C. and Rollo, M.E. and McCaffrey, T.A. and De Vlieger, N. and Van der Bend, D. and Truby, H. and Perez-Cueto, F.J.A.},
  title = {Nudging Consumers Towards Healthier Choices: A Systematic Review of Positional Influences on Food Choice},
  journal = {British Journal of Nutrition},
  year = {2016},
  volume = {115},
  pages = {2252-2263}
}

@unpublished{card23,
  title={Understanding Expert Choices Using Decision Time},
  author={Card, David and DellaVigna, Stefano and Jiang, Chenxi and Taubinsky, Dmitry},
  year={2024},
  note={Mimeo}
}

@unpublished{castillo24,
  author= {Marco Castillo and Valon Vitaku},
  title= {Salience in Choice Under Risk: An Experimental Investigation},
  year= {2025},
  note={Mimeo}
}

@article{Chabrisetal09,
  title={The Allocation of Time in Decision-Making},
  author={Chabris, Christopher F. and Morris, Carrie L. and Taubinsky, Dmitry and Laibson, David and Schuldt, Jonathon P.},
  journal={Journal of the European Economic Association},
  volume={7},
  number={2-3},
  pages={628--637},
  year={2009}
}

@unpublished{cheng2025bullet,
  title={Do Bullet Chats Keep Viewers Watching? Estimating Heterogeneous Treatment Effects with a Control Function Approach},
  author={Cheng, Yan and Wang, Jingbo and Cao, Xinyu and Shen, Zuo-Jun Max},
  note={Mimeo},
  year={2025}
}

@Article{clithero18predict,
  title={Improving Out-Of-Sample Predictions Using Response Times
and a Model of the Decision Process},
  author={Clithero, John A.},
  journal={Journal of Economic Behavior and Organization},
  volume={148},
  pages={344--375},
  year={2018}
}

@article{coffman2017size,
  title={The Size of the LGBT Population and the Magnitude of Antigay Sentiment Are Substantially Underestimated},
  author={Coffman, Katherine B. and Coffman, Lucas C. and Ericson, Keith M. Marzilli},
  journal={Management Science},
  volume={63},
  number={10},
  pages={3168--3186},
  year={2017},
  publisher={INFORMS}
}

@article{cotet2025deliberation,
  title={Deliberation During Online Bargaining Reveals Strategic Information},
  author={Cotet, Miruna and Zhao, Wenjia Joyce and Krajbich, Ian},
  journal={Proceedings of the National Academy of Sciences},
  volume={122},
  number={7},
  pages={e2410956122},
  year={2025},
  publisher={National Academy of Sciences}
}

@book{cowell2011measuring,
  title={Measuring Inequality},
  author={Cowell, Frank Alan},
  year={2011},
  publisher={Oxford University Press}
}

@incollection{daspre02,
  author    = {d'Aspremont, Claude and Gevers, Louis},
  title     = {Social Welfare Functionals and Interpersonal Comparability},
  booktitle = {Handbook of Social Choice and Welfare},
  publisher = {Elsevier},
  year      = {2002},
  editor    = {Arrow, K.J. and Sen, A. K. and Suzumura, K.},
  volume    = {1},
  chapter   = {10},
  pages     = {459--541}
}

@article{diamond1974increases,
  title={Increases in Risk and in Risk Aversion},
  author={Diamond, Peter and Stiglitz, Joseph},
  journal={Journal of Economic Theory},
  volume={8},
  number={3},
  pages={337--360},
  year={1974},
  publisher={Elsevier}
}

@article{dimaggio1996have,
  title={Have American's Social Attitudes Become More Polarized?},
  author={DiMaggio, Paul and Evans, John and Bryson, Bethany},
  journal={American Journal of Sociology},
  volume={102},
  number={3},
  pages={690--755},
  year={1996},
  publisher={University of Chicago Press}
}

@article{dutta2013inequality,
  title={Inequality of Happiness in the US: 1972--2010},
  author={Dutta, Indranil and Foster, James},
  journal={Review of Income and Wealth},
  volume={59},
  number={3},
  pages={393--415},
  year={2013},
  publisher={Wiley Online Library}
}

@article{Easterlin2005,
  title={Diminishing Marginal Utility of Income? Caveat Emptor},
  author={Easterlin, Richard A.},
  journal={Social Indicators Research},
  volume={70},
  number={3},
  pages={243--255},
  year={2005},
}

@article{echeniquesaito17,
  title={Response Time and Utility},
  author={Echenique, Federico and Saito, Kota},
  journal={Journal of Economic Behavior \& Organization},
  volume={139},
  pages={49--59},
  year={2017},
}

@article{evans2003have,
  title={Have Americans' Attitudes Become More Polarized? -- An Update},
  author={Evans, John H},
  journal={Social Science Quarterly},
  volume={84},
  number={1},
  pages={71--90},
  year={2003},
  publisher={Wiley Online Library}
}

@article{fan2014copulas,
  title={Copulas in Econometrics},
  author={Fan, Yanqin and Patton, Andrew J},
  journal={Annual Review of Economics},
  volume={6},
  number={1},
  pages={179--200},
  year={2014},
  publisher={Annual Reviews}
}

@article{fazio86,
  author={Fazio, Russell H. and Williams, Carol J.},
  title={Attitude Accessibility as a Moderator of the Attitude-Perception and Attitude-Behavior Relations: An Investigation of the 1984 Presidential Election},
  journal={Journal of Personality and Social Psychology},
  volume={51},
  number={3},
  year={1986},
  pages={505--514}
}

@incollection{fazio90,
  author={Fazio, Russell H.},
  title={A Practical Guide to the Use of Response Latency in Social Psychological Research},
  editor={Hendrick, Clyde and Clark, Margaret S.},
  booktitle={Research Methods in Personality and Social Psychology},
  publisher={Sage Publications},
  year={1990},
  pages={74--97}
}

@article{ferrando07a,
  title={An Item Response Theory Model for Incorporating Response Time Data in Binary Personality Items},
  author={Ferrando, Pere J. and Lorenzo-Seva, Urbano},
  journal={Applied Psychological Measurement},
  volume={31},
  number={6},
  pages={525--543},
  year={2007}
}

@article{ferrando07b,
  title={A Measurement Model for Likert Responses That Incorporates Response Time},
  author={Ferrando, Pere J. and Lorenzo-Seva, Urbano},
  journal={Multivariate Behavioral Research},
  volume={42},
  number={4},
  pages={675--706},
  year={2007}
}

@article{fudenberg2018speed,
  title={Speed, Accuracy, and the Optimal Timing of Choices},
  author={Fudenberg, Drew and Strack, Philipp and Strzalecki, Tomasz},
  journal={American Economic Review},
  volume={108},
  number={12},
  pages={3651--3684},
  year={2018},
  publisher={American Economic Association 2014 Broadway, Suite 305, Nashville, TN 37203}
}

@article{goldinreck19,
  title={The Analysis of Survey Data with Framing Effects},
  author={Goldin, Jacob and Reck, Daniel},
  journal={The American Statistician},
  volume={73},
  number={3},
  pages={264--272},
  year={2019},
}

@article{goldinreck20,
  title={Revealed-Preference Analysis with Framing Effects},
  author={Goldin, Jacob and Reck, Daniel},
  journal={Journal of Political Economy},
  volume={128},
  number={7},
  pages={2759--2795},
  year={2020},
}

@unpublished{goncalves24,
  title={Speed, Accuracy, and Complexity},
  author={Gon\c{c}alves, Duarte},
  year={2024},
  note={Mimeo}
}

@article{haile2008empirical,
  title={On the Empirical Content of Quantal Response Equilibrium},
  author={Haile, Philip A. and Horta\c{c}su, Ali and Kosenok, Grigory},
  journal={American Economic Review},
  volume={98},
  number={1},
  pages={180--200},
  year={2008},
  publisher={American Economic Association},
  doi={10.1257/aer.98.1.180},
  url={http://www.aeaweb.org/articles.php?doi=10.1257/aer.98.1.180}
}

@article{hammond1974simplifying,
  title={Simplifying the Choice Between Uncertain Prospects Where Preference Is Nonlinear},
  author={Hammond, John S},
  journal={Management Science},
  volume={20},
  number={7},
  pages={1047--1072},
  year={1974},
  publisher={INFORMS}
}

@unpublished{haugh2016,
    author ={Haugh, Martin} ,
    title = {An Introduction to Copulas},
    note = {Lecture Notes, \url{www.columbia.edu/~mh2078/QRM/Copulas.pdf}},
    year={2016},
}

@article{hefti2022preferences,
  title={Preferences, Confusion and Competition},
  author={Hefti, Andreas and Liu, Shuo and Schmutzler, Armin},
  journal={The Economic Journal},
  volume={132},
  number={645},
  pages={1852--1881},
  year={2022},
  publisher={Oxford University Press}
}

@article{huckfeldt99,
  title={Accessibility and the Political Utility of Partisan and Ideological Orientations},
  author={Huckfeldt, Robert and Levine, Jeffrey and Morgan, William and Sprague, John},
  journal={American Journal of Political Science},
  volume={43},
  number={3},
  pages={888--911},
  year={1999}
}

@article{johnson2004,
  title={Timepieces: Components of Survey Question Response Latencies},
  author={Johnson, Martin},
  journal={Political Psychology},
  volume={25},
  number={5},
  pages={679--702},
  year={2004}
}

@article{johnson2006simple,
  title={On the Simple Economics of Advertising, Marketing, and Product Design},
  author={Johnson, Justin P and Myatt, David P},
  journal={American Economic Review},
  volume={96},
  number={3},
  pages={756--784},
  year={2006},
  publisher={American Economic Association}
}

@article{KaiserOswald22,
  title={Inequality, Well-Being, and the Problem of the Unknown Reporting Function},
  author={Kaiser, Caspar and Oswald, Andrew J.},
  journal={Proceedings of the National Academy of Sciences},
  volume={119},
  number={50},
  pages={e2217750119},
  year={2022}
}

@article{kalmijn2005measuring,
  title={Measuring Inequality of Happiness in Nations: In Search for Proper Statistics},
  author={Kalmijn, Wim and Veenhoven, Ruut},
  journal={Journal of Happiness Studies},
  volume={6},
  number={4},
  pages={357--396},
  year={2005},
  publisher={Springer}
}

@article{kellogg1931,
  title={The Time of Judgment in Psychometric Measures},
  author={Kellogg, W. N.},
  journal={American Journal of Psychology},
  volume={43},
  number={1},
  pages={65-86},
  year={1931}
}

@book{kendall1955rank,
  title={Rank Correlation Methods},
  author={Kendall, Maurice George},
  year={1955},
  publisher={New York:
Hafner Publishing Co}
}

@article{kiefer1988economic,
  title={Economic Duration Data and Hazard Functions},
  author={Kiefer, Nicholas M},
  journal={Journal of Economic Literature},
  volume={26},
  number={2},
  pages={646--679},
  year={1988},
  publisher={JSTOR}
}

@article{Kitagawa2018ECTA,
  author  = {Kitagawa, Toru and Tetenov, Aleksey},
  title   = {Who Should Be Treated? Empirical Welfare Maximization Methods for Treatment Choice},
  journal = {Econometrica},
  year    = {2018},
  volume  = {86},
  number  = {2},
  pages   = {591--616},
  doi     = {10.3982/ECTA13288},
  url     = {https://doi.org/10.3982/ECTA13288}
}

@article{konovalov2019revealed,
  title={Revealed Strength of Preference: Inference From Response Times},
  author={Konovalov, Arkady and Krajbich, Ian},
  journal={Judgment \& Decision Making},
  volume={14},
  number={4},
  pages={381--394},
  year={2019}
}

@article{krajbich2012attentional,
  title={The Attentional Drift-Diffusion Model Extends to Simple Purchasing Decisions},
  author={Krajbich, Ian and Lu, Dingchao and Camerer, Colin and Rangel, Antonio},
  journal={Frontiers in Psychology},
  volume={3},
  pages = {193},
  year={2012},
  publisher={Frontiers Media},
  doi={10.3389/fpsyg.2012.00193},
  url={https://doi.org/10.3389/fpsyg.2012.00193}
}

@article{Krefeld2024PNAS,
  author  = {Krefeld-Schwalb, Antonia and Sugerman, Eli Rosen and Johnson, Eric J.},
  title   = {Exposing Omitted Moderators: Explaining Why Effect Sizes Differ in the Social Sciences},
  journal = {Proceedings of the National Academy of Sciences},
  year    = {2024},
  volume  = {121},
  number  = {12},
  pages   = {e2306281121},
  doi     = {10.1073/pnas.2306281121},
  url     = {https://www.pnas.org/doi/10.1073/pnas.2306281121}
}

@article{lelkes2016,
  title={Mass Polarization: Manifestations and Measurements},
  author={Lelkes, Yptach},
  journal={Public Opinion Quarterly},
  volume={80},
  number={S1},
  pages={392--410},
  year={2016}
}

@article{liunetzer23,
  title={Happy Times: Measuring Happiness Using Response Times},
  author={Liu, Shuo and Netzer, Nick},
  journal={American Economic Review},
  volume={113},
  pages={3289--3322},
  year={2023}
}

@unpublished{liunetzer20,
  title={Happy Times: Identification From Ordered Response Data},
  author={Liu, Shuo and Netzer, Nick},
  note={University of Zurich, Department of Economics, Working Paper No. 371},
  year={2020},
}

@book{measurementbook,
  title={Foundations of Measurement, Vol. III: Representation, Axiomatization, and Invariance},
  author={Luce, R. Duncan and Krantz, David H. and Suppes, Patrick and Amos Tversky},
  year={1990},
  publisher={Academic Press}
}

@article{madrian01,
  title={The Power of Suggestion: Inertia in 401(k) Participation and Savings Behavior},
  author={Madrian, Brigitte C. and Shea, Dennis F.},
  journal={Quarterly Journal of Economics},
  volume={116},
  year={2001},
  number={4},
  pages={1149--1187}
}

@ARTICLE{Manski88,
  author = {Manski, Charles F.},
  title = {Identification of Binary Response Models},
  journal = {Journal of the American Statistical Association},
  year = {1988},
  volume = {83},
  pages = {729--738},
  number = {403}
}

@article{Manski2004Statistical,
  author    = {Charles F. Manski},
  title     = {Statistical Treatment Rules for Heterogeneous Populations},
  journal   = {Econometrica},
  year      = {2004},
  volume    = {72},
  number    = {4},
  pages     = {1221--1246},
  publisher = {The Econometric Society},
  month     = jul,
  url       = {https://www.jstor.org/stable/3598783}
}

@article{maskin2000asymmetric,
  title={Asymmetric Auctions},
  author={Maskin, Eric and Riley, John},
  journal={The Review of Economic Studies},
  volume={67},
  number={3},
  pages={413--438},
  year={2000},
  publisher={Wiley-Blackwell}
}

@article{matzkin1992nonparametric,
  title={Nonparametric and Distribution-Free Estimation of the Binary Threshold Crossing and The Binary Choice Models},
  author={Matzkin, Rosa L.},
  journal={Econometrica},
  volume={60},
  number={2},
  pages={239--270},
  year={1992},
  publisher={The Econometric Society},
  doi={10.2307/2951596},
  url={https://doi.org/10.2307/2951596}
}

@article{milgrom1981good,
  title={Good News and Bad News: Representation Theorems and Applications},
  author={Milgrom, Paul R},
  journal={The Bell Journal of Economics},
  pages={380--391},
  volume={12},
  number={2},
  year={1981},
  publisher={JSTOR}
}

@ARTICLE{Moffatt05,
  author = {Moffatt, Peter G.},
  title = {Stochastic Choice and the Allocation of Cognitive Effort},
  journal = {Experimental Economics},
  year = {2005},
  volume = {8},
  pages = {369--388},
  number = {4}
}

@article{moss2023ethical,
  title={Is It Ethical to Use Mechanical Turk for Behavioral Research? Relevant Data From a Representative Survey of {MT}urk Participants and Wages},
  author={Moss, Aaron J and Rosenzweig, Cheskie and Robinson, Jonathan and Jaffe, Shalom N and Litman, Leib},
  journal={Behavior Research Methods},
  volume={55},
  number={8},
  pages={4048--4067},
  year={2023},
  publisher={Springer}
}

@article{moyerbayer1976,
  title={Mental Comparison and the Symbolic Distance Effect},
  author={Moyer, Robert S. and Bayer, Richard H.},
  journal={Cognitive Psychology},
  volume={8},
  number={2},
  pages={228--246},
  year={1976}
}

@article{Oswald08,
  title={On the Curvature of the Reporting Function From Objective Reality to Subjective Feelings},
  author={Oswald, Andrew J.},
  journal={Economics Letters},
  year={2008},
  pages = {369-372},
  volume={100},
  number={3}
}

@article{palmer2005,
  title={The Effect of Stimulus Strength on the Speed and Accuracy of a Perceptual Decision},
  author={Palmer, John and Huk, Alexander C. and Shadlen, Michael N.},
  journal={Journal of Vision},
  volume={5},
  pages={376--404},
  year={2005}
}

@article{QuahStrulovici2012,
  author    = {John K.-H. Quah and Bruno Strulovici},
  title     = {Aggregating the Single Crossing Property},
  journal   = {Econometrica},
  volume    = {80},
  number    = {5},
  pages     = {2333--2348},
  year      = {2012},
  publisher = {The Econometric Society},
  url       = {https://www.jstor.org/stable/23271450}
}

@article{ramsey27,
  title={A Contribution to the Theory of Taxation},
  author={Ramsey, F.P.},
  journal={The Economic Journal},
  volume={37},
  number={145},
  pages={47--61},
  year={1927},
}

@article{ranger13,
  title={Modeling Responses and Response Times in Personality Tests With Rating Scales},
  author={Ranger, Jochen},
  journal={Psychological Test and Assessment Modeling},
  volume={55},
  number={4},
  pages={361--382},
  year={2013},
}

@article{romanowolf2016,
  title={Efficient Computation of Adjusted p-Values for Resampling-Based Stepdown Multiple Testing},
  author={Romano, Joseph P. and Wolf, Michael},
  journal={Statistics and Probability Letters},
  volume={113},
  pages={38--40},
  year={2016},
}

@article{rubinstein2007instinctive,
  title={Instinctive and Cognitive Reasoning: A Study of Response Times},
  author={Rubinstein, Ariel},
  journal={The Economic Journal},
  volume={117},
  pages={1243--1259},
  year={2007},
  publisher={Royal Economic Society, Blackwell Publishing},
  doi={10.1111/j.1468-0297.2007.02081.x},
  url={https://doi.org/10.1111/j.1468-0297.2007.02081.x}
}

@article{rubinstein2013response,
  title={Response Time and Decision Making: An Experimental Study},
  author={Rubinstein, Ariel},
  journal={Judgment and Decision Making},
  volume={8},
  number={5},
  pages={540--551},
  year={2013},
  doi={10.1017/S1930297500003648}
}

@article{rubinstein2016typology,
  title={A Typology of Players: Between Instinctive and Contemplative},
  author={Rubinstein, Ariel},
  journal={The Quarterly Journal of Economics},
  volume={131},
  number={2},
  pages={859--890},
  year={2016},
  publisher={Oxford University Press}
}

@article{samuelson38,
  title={A Note on the Pure Theory of Consumer's Behaviour},
  author={Samuelson, Paul A.},
  journal={Economica},
  volume={5},
  number={17},
  pages={61--71},
  year={1938},
}

@article{serragarcia23,
  title={Incentives and Defaults Can Increase COVID-19 Vaccine Intentions and Test Demand},
  author={Serra-Garcia, Marta and Szech, Nora},
  journal={Management Science},
  volume={69},
  number={2},
  pages={1037--1049},
  year={2023},
}

@book{shaked2007stochastic,
  title={Stochastic Orders},
  author={Shaked, Moshe and Shanthikumar, J. George},
  year={2007},
  publisher={Springer}
}

@article{shevlin22,
  title={High-Value Decisions Are Fast and Accurate, Inconsistent With Diminishing Value Sensitivity},
  author={Shevlin, Blair R.K. and Smith, Stephanie M. and Hausfeld, Jan and Krajbich, Ian},
  journal={Proceedings of the National Academy of Sciences},
  volume={119},
  number={6},
  pages={e2101508119},
  year={2022}
}

@article{spearman1904proof,
  title={The Proof and Measurement of Association Between Two Things},
  author={Spearman, Charles},
  journal={The American Journal of Psychology},
  volume={15},
  number={1},
  pages={72--101},
  year={1904},
  publisher={JSTOR}
}

@article{stevenson2008happiness,
  title={Happiness Inequality in the United States},
  author={Stevenson, Betsey and Wolfers, Justin},
  journal={The Journal of Legal Studies},
  volume={37},
  number={S2},
  pages={S33--S79},
  year={2008},
  publisher={The University of Chicago Press}
}

@article{takane83,
  title={Multidimensional Scaling Models for Reaction Times and Same-Different Judgements},
  author={Takane, Yoshio and Sergent, Justine},
  year={1983},
  journal={Psychometrika},
  volume={48},
  pages={393--423},
}

@ARTICLE{thalersun03,
  author = {Thaler, R.H. and Sunstein, C.R.},
  title = {Libertarian Paternalism},
  journal = {American Economic Review, Papers and Proceedings},
  year = {2003},
  volume = {93},
  pages = {175-179}
}

@BOOK{thalersun08,
  author = {Thaler, R.H. and Sunstein, C.R.},
  title = {Nudge: Improving Decisions About Health, Wealth, and Happiness},
  publisher = {New Haven: Yale University Press},
  year = {2008}
}

@article{thistle1989ranking,
  title={Ranking Distributions with Generalized Lorenz Curves},
  author={Thistle, Paul D},
  journal={Southern Economic Journal},
  pages={1--12},
  volume    = {56},
  number    = {1},
  year={1989},
  publisher={JSTOR}
}

@article{Tversky72,
	author= {Tversky, Amos},
	title= {Elimination by Aspects: A Theory of Choice},
	journal= {Psychological Review},
	volume={79},
    number={4},
	pages={281--299},
	year={1972}	
}

@article{Tversky81,
	author= {Tversky, Amos and Kahneman, Daniel},
	title= {The Framing of Decisions and the Psychology of Choice},
	journal= {Science},
	volume={211},
    number={4481},
	pages={453--458},
	year={1981}	
}

@unpublished{vaeth23,
  title={Voter Learning, Unidimensional Ideology, and Polarization},
  author={Vaeth, Martin},
  year={2023},
  note={Mimeo}
}

@article{wang2024weighted,
  title={Weighted Utility and Optimism/Pessimism: A Decision-Theoretic Foundation of Various Stochastic Dominance Orders},
  author={Wang, Tao and Lehrer, Ehud},
  journal={American Economic Journal: Microeconomics},
  volume={16},
  number={1},
  pages={210--223},
  year={2024},
  publisher={American Economic Association 2014 Broadway, Suite 305, Nashville, TN 37203-2425}
}

\newpage

\section*{Appendices}

\begin{appendix}

\section{Rationalizability}\label{App.Rat}

In this appendix, we study the question of \textit{rationalizability} of the data:\ when $\mathscr{G}^*$ and $\mathscr{C}^*$ are both restricted sets of distributions and chronometric functions, respectively, does there exist $((G_j)_j,(c_j)_j) \in \mathscr{G}^* \times \mathscr{C}^*$ that induces the given data $(p_j,F_j)_j$?

We give an answer to this question for the case in which both $\mathscr{G}^*$ and $\mathscr{C}^*$ are generated using sets of transformations with properties like before. Formally, let $\mathscr{C}^*$ be the \textit{maximal set} generated by $((c_j^*)_j,\Psi)$, that is,
\[\mathscr{C}^* = \{ (c_j)_j \in \mathscr{C} \mid (c_j)_j = (c_j^* \circ \psi_j)_j \text{ for some } (\psi_j)_j \in \Psi\}.\] 
In words, not only can every $(c_j)_j \in \mathscr{C}^*$ be derived from $(c_j^*)_j$ using a composition with some $(\psi_j)_j \in \Psi$, but the composition of $(c_j^*)_j$ with any $(\psi_j)_j \in \Psi$ yields an element of $\mathscr{C}^*$. This requires that $\psi_j(0)=0$ holds for all transformations in $\Psi$. Analogously, let
\[\mathscr{G}^* = \{ (G_j)_j \in \mathscr{G} \mid (G_j)_j = (G_j^* \circ \phi_j)_j \text{ for some } (\phi_j)_j \in \Phi\}\] 
be the maximal set generated by $((G_j^*)_j,\Phi)$. For example, if $(G^*_j)_j$ are strictly increasing cdfs and $\Phi$ is the set of all profiles of strictly increasing transformations, then $\mathscr{G}^*$ becomes the set of all profiles of strictly increasing cdfs. If, in addition, $(G^*_j)_j$ are symmetric around their means, then combined with a suitably defined set of symmetric transformations we obtain a set of profiles of distributions that are also symmetric.

For the following result, denote by $\Phi \circ \Psi^{-1}=\{ (\phi_j \circ \psi_j^{-1})_j \mid (\phi_j) _j \in \Phi \text{ and } (\psi)_j \in \Psi \}$ the set of all profiles of functions which are compositions of the elements of $\Phi$ and the inverse elements of $\Psi$. Any such function is bijective and strictly increasing, hence continuous.

\begin{theorem}\label{t2} Let $\mathscr{C}^*$ and $\mathscr{G}^*$ be the maximal sets generated by $((c^*_j)_j,\Psi)$ and $((G^*_j)_j,\Phi)$, respectively. Then, data $(p_j,F_j)_j$ is rationalizable if and only if there exists $(L_j)_j \in \Phi \circ \Psi^{-1}$ such that $(G_j^* \circ L_j)_j=(H_j)_j$.
\end{theorem}

\begin{proof} \textit{Only-if-statement.} Suppose there exists $((G_j)_j,(c_j)_j) \in \mathscr{G}^* \times \mathscr{C}^*$ that induces $(p_j,F_j)_j$. As shown in the proof of Theorem \ref{t1}, there exists $(\psi_j)_j \in \Psi$ such that $(H_j)_j=(G_j \circ \psi^{-1}_j)_j$, because $((c^*_j)_j,\Psi)$ generates $\mathscr{C}^*$. It also holds that there exists $(\phi_j)_j \in \Phi$ such that $(G_j)_j=(G^*_j \circ \phi_j)_j$, because $((G^*_j)_j,\Phi)$ generates $\mathscr{G}^*$. Now define $L_j=(\phi_j \circ \psi^{-1}_j)$ for all $j \in J$, so that $(L_j)_j \in \Phi \circ \Psi^{-1}$ holds. Furthermore, \[(G^*_j \circ L_j)_j = (G^*_j \circ \phi_j \circ \psi^{-1}_j)_j = (G_j \circ \psi^{-1}_j)_j=(H_j)_j.\]

\textit{If-statement.} Suppose there exists $(L_j)_j \in \Phi \circ \Psi^{-1}$ that satisfies $(G_j^* \circ L_j)_j=(H_j)_j$. Let $(\phi_j)_j \in \Phi$ and $(\psi_j)_j \in \Psi$ be such that $(L_j)_j=(\phi_j \circ \psi^{-1}_j)_j$. Then $(G_j)_j:=(G^*_j \circ \phi_j)_j \in \mathscr{G}^*$ because $\mathscr{G}^*$ is the maximal set generated by $((G^*_j)_j,\Phi)$, and $(c_j)_j:=(c^*_j \circ \psi_j)_j \in \mathscr{C}^*$ because $\mathscr{C}^*$ is the maximal set generated by $((c^*_j)_j,\Psi)$. The data induced by $((G_j)_j,(c_j)_j)$ are, for all $t \in [\underline{t},\overline{t}]$ and $j \in J$,
\begin{align*}\hat{p}^0_j \hat{F}^0_j(t) & =G_j((c^0_j)^{-1}(t)) \\ & =G^*_j(\phi_j((c^0_j)^{-1}(t))) \\ & =G^*_j(\phi_j(\psi^{-1}_j((c_j^{*,0})^{-1}(t))))\\ & = G^*_j(L_j((c_j^{*,0})^{-1}(t)))\\ & = H_j((c_j^{*,0})^{-1}(t))\\ & = p^0_j F^0_j(t),\end{align*} where the third equality has been established in the proof of Theorem \ref{t1}. The analogous argument shows that $((G_j)_j,(c_j)_j)$ induces also $p^1_j F^1_j(t)$, hence the data is rationalizable.\end{proof}

For rationalizability, we need to check whether there exist functions $(L_j)_j \in \Phi \circ \Psi^{-1}$ which satisfy the implicit condition
\[G^*_j(L_j(x)) = H_j(x) \text{ for all } x \in \mathbb{R}.\] 
This is easier than it may appear. For example, if both $H_j$ and $G^*_j$ are strictly increasing, then the candidate $L_j$ is unique and given by $L_j=(G^*_j)^{-1} \circ H_j$, which is a bijective and strictly increasing function based on observables. It remains to be checked whether this function can be written as an admissible composition $\phi_j \circ \psi^{-1}_j$. This is always the case, for example, if $\Phi$ is the unrestricted set of all profiles of transformations, because $\Phi \circ \Psi^{-1}$ is unrestricted in that case as well. In other cases, the function $L_j$ will be unique in some intervals and can be extended outside these intervals in a way that guarantees bijectivity and monotonicity. One then needs to check whether $L_j$ coincides with an admissible composition $\phi_j \circ \psi^{-1}_j$ wherever it is uniquely defined.

Other cases are even easier. For example, if $H_j$ takes the value $0$ for finite values $x<0$ but $G_j^*$ does not, then a function $L_j$ satisfying $G^*_j \circ L_j = H_j$ cannot be strictly increasing and hence cannot be a composition $\phi_j \circ \psi^{-1}_j$. This implies that the data is not rationalizable. Intuitively, since $G^*_j$ extends to infinity and the data has no atoms at response time $\underline{t}$, it can only be rationalized by a chronometric function that satisfies $c_j(x)>\underline{t}$ for all $x<0$, and hence $H_j$ cannot reach $0$. 

\section{Robustness to Heterogeneity and Noise}\label{App.Stoch}

As described in the main text, we model response time in the decision problem of interest as $t=c_j(x) \cdot \eta \cdot \epsilon$ and in the baseline problem as $t_b = \phi \cdot \eta$, where $\eta > 0$ captures the individual's general response speed, $\epsilon > 0$ comprises different sources of noise, and $\phi>0$ parameterizes the baseline problem. The terms $x$ and $\eta$ may be correlated and follow a $j$-specific distribution. The cdf of the marginal distribution of $x$ is denoted by $G_j(x)$ and is assumed to be continuous. The terms $\epsilon$ and $\phi$ may also be correlated and follow a $j$-specific distribution described by a joint probability measure $\mu_j$, but they are assumed to be independent of the other variables.

In a slight deviation from the main model, we let $\underline{t}=0$ and $\overline{t}=\infty$, as heterogeneity and noise may spread out response times to arbitrary values. The chronometric function $c_j: \mathbb{R} \rightarrow [0,\infty)$ is described by $c_j^0: (-\infty,0) \rightarrow [0,\infty)$ and $c_j^1: (0,+\infty) \rightarrow [0,\infty)$, assumed to be continuous, strictly increasing/decreasing throughout, and to approach the respective limits asymptotically. Continuity of $G_j(x)$ implies that $x=0$ is a zero probability event and allows us to leave $c_j(0)$ unspecified.

We consider normalized response times $\hat{t}= t/ t_b = c_j(x) \cdot \epsilon / \phi $ and denote the cdfs of their distributions (conditional on choice) by $\hat{F}_j^i$. The model $(G_j,c_j,\mu_j)_j$ then generates the data
\[p_j^0 \hat{F}_j^0(t)= \int_{\text{supp } \mu_j} G_j((c_j^0)^{-1}(t \phi / \epsilon)) \, d \mu_j(\phi,\epsilon)\] and
\[p_j^1 \hat{F}_j^1(t)= 1 - \int_{\text{supp } \mu_j} G_j((c_j^1)^{-1}(t \phi / \epsilon)) \, d \mu_j(\phi,\epsilon),\] for all $t \geq 0$ and $j \in J$.

Suppose the analyst relies on the representative chronometric function $c_j^*(x)=1/|x|$ and therefore constructs the empirical distributions \begin{equation}\label{defHhat}\hat{H}_j(x) = \begin{cases} 1 - p^1_j \hat{F}^1_j(1/x) & \text{if } x>0, \\ p_j^0 & \text{if } x=0, \\p^0_j \hat{F}^0_j(-1/x) & \text{if } x < 0,\end{cases}\end{equation} from those data, for all $j \in J$. The following result shows how these distributions, which are misspecified as they ignore heterogeneity and noise, are related to the model fundamentals.

\begin{proposition}\label{apppropnoise} Given $(G_j,c_j,\mu_j)_j$, the distributions $(\hat{H}_j)_j$ defined in (\ref{defHhat}) can be written as \begin{equation}\label{resHhat}\hat{H}_j(x)= \int_{\text{supp } \mu_j} K_j(x \epsilon/\phi) \, d \mu_j(\phi,\epsilon),\end{equation} where $(K_j)_j = (G_j \circ \psi_j^*)_j$ with \begin{equation*}\psi_j^*(x) = \begin{cases} (c_j^1)^{-1}(1/x)& \text{if } x>0, \\ 0 & \text{if } x=0, \\ (c^0_j)^{-1}(-1/x) & \text{if } x < 0,\end{cases}\end{equation*} for all $j \in J$.\end{proposition}

\begin{proof} Consider any $j \in J$. For $x=0$ we have
\begin{align*}\int_{\text{supp } \mu_j} K_j(x \epsilon/\phi) \, d \mu_j(\phi,\epsilon) =  K_j(0) =G_j(0)=p_j^0=\hat{H}_j(0).\end{align*}
For any $x < 0$ we have
\begin{align*}\int_{\text{supp } \mu_j} K_j(x \epsilon/\phi) \, d \mu_j(\phi,\epsilon) & = \int_{\text{supp } \mu_j} G_j((c_j^0)^{-1}(-(1/x)(\phi/\epsilon))) \, d \mu_j(\phi,\epsilon) \\ & = p_j^0 \hat{F}_j^0(-1/x) = \hat{H}_j(x).\end{align*} For any $x > 0$ we analogously have
\begin{align*}\int_{\text{supp } \mu_j} K_j(x \epsilon/\phi) \, d \mu_j(\phi,\epsilon) & = \int_{\text{supp } \mu_j} G_j((c_j^1)^{-1}((1/x)(\phi/\epsilon))) \, d \mu_j(\phi,\epsilon) \\ & = 1 - p_j^1 \hat{F}_j^1(1/x)  = \hat{H}_j(x),\end{align*} which establishes (\ref{resHhat}). \end{proof}

Proposition \ref{apppropnoise} extends the main insight from the proof of Theorem \ref{t1} to the case with noise. Indeed, if there is no noise (i.e., $\epsilon/\phi$ is fixed at $1$), then it implies that $\hat{H}_j = K_j = G_j \circ \psi_j^*$, where $\psi_j^*$ corresponds to $\psi_j^{-1}$ in the proof of Theorem \ref{t1}. 

In the general case, the result can be applied to check whether $(\hat{H}_j)_j$ is suitable for testing the hypothesis that $(G_j)_j$ satisfies a property $\mathbf{P}$. The first requirement is that $\mathbf{P}$ of $(G_j)_j$ is invariant to $\psi_j^*$, so that $(K_j)_j$ also satisfies $\mathbf{P}$. Without knowledge of the chronometric function beyond monotonicity, this requires invariance to all increasing transformations. With the knowledge of symmetry, it requires invariance to symmetric transformations, and so on. The second (and crucial) requirement is invariance of $\mathbf{P}$ of $(K_j)_j$ to the multiplicative convolutions (\ref{resHhat}), so that $(\hat{H}_j)_j$ inherits $\mathbf{P}$. 

To illustrate the argument, consider the property of a pair $(G_1,G_2)$ of distributions that $G_1(-x) + G_1(x) \leq G_2(-x) + G_2(x)$ for all $x \geq 0$, which implies that the mean of $G_1$ is larger than the mean of $G_2$ and which we use in our applications in Section \ref{sec.app}. This property is invariant to transformations that are symmetric around zero and identical for the two groups. The property is then also invariant to (\ref{resHhat}) for noise measures $\mu_j$ that are identical for the two groups. The reason is that the inequality holds pointwise for each $(\phi,\epsilon)$ in the integrals. Consequently, $(\hat{H}_1,\hat{H}_2)$ can be used to test the property under the assumption of symmetric and common chronometric functions and common noise distributions.

\section{Ranking of Means}\label{App.Rank}

In this appendix, we show that the inequality $G_A(x)+G_A(-x) \leq G_B(x) + G_B(-x)$ for all $x \in \mathbb{R}_+$ implies that the mean of $G_A$, denoted $\mu_A$, is larger than the mean of $G_B$, denoted $\mu_B$ (assuming that both means exist). Using the fact that \[\mu_j = -\int_{-\infty}^0 G_j(x) dx + \int_{0}^{+\infty} [1-G_j(x)] dx,\] 
we have
\begin{align*}
\mu_A - \mu_B = \, & \int_{-\infty}^0[G_B(x)-G_A(x)]dx + \int_{0}^{+\infty}[1-G_A(x)-1+G_B(x)]dx \\ 
= \, & \int_{0}^{+\infty}[G_B(-x)-G_A(-x)]dx + \int_{0}^{+\infty}[G_B(x)-G_A(x)]dx\\
= \, & \int_{0}^{+\infty}[G_B(x) + G_B(-x)-G_A(x)-G_A(-x)]dx \geq 0.
\end{align*} 

\section{Further Applications}\label{sec.theory.app}

This appendix develops further applications of Theorem \ref{t1} with distributional properties that arise in a variety of empirical settings. For each property, we identify the relevant invariance class and characterize the conditions under which detection is feasible using response times. Although we do not implement these applications empirically, the results provide a template for future work.

\subsection{Inequality}\label{sec:lorenz}

Our method can be used to study inequality or dispersion of a latent variable, which has applications across multiple fields. For instance, within the literature on subjective well-being, there is a substantial interest in understanding the inequality of happiness \citep[e.g.,][]{stevenson2008happiness}. Similarly, researchers have studied societal polarization by measuring the dispersion of individual attitudes towards social and political issues \citep{dimaggio1996have, evans2003have}. In the context of market competition, the spread of consumer preferences has direct implications for the optimal pricing and advertising strategies of firms \citep{johnson2006simple,hefti2022preferences}. Unfortunately, standard measures of inequality like the Gini index require cardinal information, which makes their application to ordered response data questionable \citep[see the discussion in][]{dutta2013inequality}. Response times may serve as a source of cardinal information even when decisions are binary, such as whether to buy or not buy a product \citep{cotet2025deliberation}.

The Lorenz curve is a convenient graphical representation of a distribution's inequality, and interesting measures of inequality like the Gini index are based on the Lorenz curve \citep{atkinson1970measurement,cowell2011measuring}. For any distribution $G$, the Lorenz curve is defined by
\[L(q, G)=\frac{\int_0^q G^{-1}(x)dx}{\int_0^1 G^{-1}(x)dx} \; \text{ for all } q \in [0,1],\]
where $G^{-1}(x):= \inf\{y \mid G(y) \geq x\}$ denotes the left inverse of $G$. In the context of subjective well-being, $L(q, G)$ could be understood as the proportion of total happiness allocated to the least happy $100q$ percent of the population. How far the Lorenz curve falls below the 45-degree line is an indication of how unequal the distribution is. The curve is invariant to all linear transformations of the distribution $G$. Therefore, if we plot the Lorenz curve for an empirical function $H$ constructed as in (\ref{defH}) with any representative function $c^*$, exactly this Lorenz curve (and any measure based on it like the Gini index) is detected for the class of chronometric functions that are linearly generated from $c^*$. One can repeat this procedure for various different representative functions $c^{k*}$ to obtain bounds on the true Lorenz curve for larger sets of chronometric functions.


Beyond measuring inequality within a single population, some applications require comparing inequality across groups. The subjective well-being literature has been interested in how inequality of happiness changed over time and across nations \citep{kalmijn2005measuring, stevenson2008happiness,dutta2013inequality}. There is also a great interest in political polarization trends \citep{dimaggio1996have, evans2003have}. Furthermore, firms have an interest in learning about changes in the spread of consumer preferences, to optimally adapt their pricing and advertising strategies  \citep{johnson2006simple,hefti2022preferences}. Our method applies to such comparative analyses as well.

Consider first the problem of detecting whether $G_1$ has a smaller variance than $G_2$, which is a convenient way of ranking the dispersion of distributions due to its ability to provide a complete order. This property is invariant only to transformations $(\psi_1, \psi_2)$ that are linear and additionally satisfy $\psi_1=\psi_2$, which may be too demanding. Comparing instead the inequality based on Lorenz curves has the benefit of avoiding direct scale comparisons, as the two distributions' Lorenz curves are invariant to transformations $(\psi_1, \psi_2)$ that are linear but not necessarily identical. For example, we can say that $G_1$ Lorenz-dominates $G_2$ if $L(q, G_1)\geq L(q, G_2)$ holds for all $q\in [0,1]$ \citep[see][]{shaked2007stochastic}. Then, under the assumption of knowing the chronometric functions up to group-specific linear transformations, detecting a Lorenz-dominance relationship between $G_1$ and $G_2$ (or the absence thereof) is equivalent to checking whether the property holds for the respective empirical functions $(H_1, H_2)$.

Another measure that can be used for assessing  the relative dispersion of distributions is the concept of single-crossing dominance \citep{diamond1974increases,hammond1974simplifying}. Formally, $G_1$ single-crossing dominates $G_2$ if there exists $x^\ast$ such that $G_1(x) \leq G_2(x)$ if $x \leq x^\ast$ and $G_1(x) \geq G_2(x)$ if $x \geq x^\ast$. Intuitively, the condition reflects that $G_1$ assigns less probability weight to values at the tails of the distribution compared to $G_2$. Although the property of single-crossing dominance is sensitive to group-specific transformations, it has the advantage of being robust to non-linear transformations. Just like FOSD, the property (and its violation) is invariant to all profiles of transformations $(\psi_1, \psi_2)$ that satisfy $\psi_1=\psi_2$. Therefore, according to Theorem \ref{t1}, single-crossing dominance is detected for a relatively large class of chronometric functions if $H_1$ single-crossing dominates $H_2$, where $(H_1, H_2)$ are constructed from data using \eqref{defH}, and a violation thereof is detected otherwise. The condition can again be expressed directly in terms of the data. It requires that $p_1^iF_1^i(t)$ and $p_2^iF_2^i(t)$ cross at most once for one of the two responses $i \in \{0,1\}$ and not at all for the other response.

Finally, it is known that either Lorenz-dominance or single-crossing dominance imply second-order stochastic dominance (SOSD) when $G_1$ has a higher mean than $G_2$ \citep{thistle1989ranking}. Therefore, combining several of the results derived so far allows us to detect SOSD.

\subsection{Likelihood-Ratio Dominance}

Consider the property that $G_1$ likelihood-ratio dominates $G_2$, defined by the inequality
\[
(G_1(x)-G_1(x''))(G_2(x)-G_2(x')) \leq (G_1(x)-G_1(x'))(G_2(x)-G_2(x''))
\]
for all $x'' < x' < x$ \citep[see][]{wang2024weighted}. If $G_1$ and $G_2$ are absolutely continuous, the property can be equivalently expressed in terms of their density functions $g_1$ and $g_2$ as $g_1(x')g_2(x) \leq g_1(x)g_2(x')$ for any $x' < x$ \citep[see][]{shaked2007stochastic}. Likelihood-ratio dominance, which is stronger than FOSD, has proven useful in a range of economic applications that involve monotone comparative statics under uncertainty \citep{milgrom1981good, athey2002monotone}. Suppose that the goal is to maximize the expected value of a function $\pi(y,x)$ by choice of $y$ when $x$ is distributed according to $G_j$. If $\pi(y,x)$ satisfies a single-crossing property and $G_1$ likelihood-ratio dominates $G_2$, then the optimal choice of $y$ is larger for $G_1$ than for $G_2$ \citep[under appropriate technical conditions, see][]{athey2002monotone}.  For example, a firm may introduce a new product in two different countries, each of which is characterized by a distribution of consumer types. Higher types imply a higher demand for the product in a way that the firm's profit $\pi(p,x)$ as a function of price $p$ and type $x$ is single-crossing. If the firm's market research (e.g., a survey that asks potential consumers whether they like the product) allows the detection of likelihood-ratio dominance of the type distributions of the two countries, the firm knows in which country to charge a higher price. Other applications of the concept include optimal investment decisions and bidding in auctions.

Like FOSD, likelihood-ratio dominance (and its violation) is invariant to all profiles of transformations $(\psi_1, \psi_2)$ that satisfy $\psi_1=\psi_2$. We thus proceed as for FOSD by constructing $(H_1, H_2)$ from data and then verifying whether $H_1$ likelihood-ratio dominates $H_2$. For example, checking the respective inequalities for all $x'' < x' < x \leq 0$ can be expressed in terms of the data as
\[
(p_1^0F_1^0(t)-p_1^0F_1^0(t''))(p_2^0F_2^0(t)-p_2^0F_2^0(t')) \leq (p_1^0F_1^0(t)-p_1^0F_1^0(t'))(p_2^0F_2^0(t)-p_2^0F_2^0(t''))
\]
for all $t'' < t' < t$. The full condition is particularly easy to express when $p_j^i>0$ and the response time distributions $F^i_j$ admit strictly positive densities $f^i_j$. In that case, we detect likelihood-ratio dominance if the empirical likelihood-ratio
\[\frac{p^i_1f_1^i(t)}{p^i_2f_2^i(t)}\]
is weakly increasing in $t$ for $i=0$ and weakly decreasing in $t$ for $i=1$, and
\[\frac{p^0_1f_1^0(\overline{t})}{p^0_2f_2^0(\overline{t})} \leq \frac{p^1_1f_1^1(\overline{t})}{p^1_2f_2^1(\overline{t})}\]
holds. Otherwise, a violation of likelihood-ratio dominance is detected.

Our approach extends analogously to the detection of hazard-rate dominance and reversed hazard-rate dominance, both of which lie between likelihood-ratio dominance and FOSD, and which are also commonly used in applied settings \citep[e.g.,][]{kiefer1988economic,maskin2000asymmetric,wang2024weighted}.

\subsection{Unimodality}

Consider the property that the distribution is unimodal with mode at zero, i.e., $G$ is convex below zero and concave above zero, and strictly so except when $G(x) \in \{0,1\}$. Unimodality is of interest once more in political applications where $G$ describes the distribution of political attitudes in a population. Unimodality reflects a centered population where more extreme positions receive less support. An ongoing debate in political science concerns whether political attitudes follow unimodal distributions or are polarized and better described by bimodal distributions \citep[see][]{lelkes2016,vaeth23}. Analysts often want to learn about these properties from survey responses, but as \cite{vaeth23} points out, ordinal responses are inadequate to test for properties like uni- or bimodality of the underlying distribution.

Unimodality is invariant to (sigmoid) transformations that satisfy $\psi(0)=0$ and are weakly convex below zero and weakly concave above zero. Starting from a representative chronometric function $c^*$, such transformations allow us to generate all chronometric functions which are weakly ``more convex'' than $c^*$ on $\mathbb{R}_-$ and on $\mathbb{R}_+$ (because $c^*$ is increasing on $\mathbb{R}_-$ but decreasing on $\mathbb{R}_+$). We will use this insight in a slightly different way than before. Assume that the observed cdfs $F^i$ are strictly increasing on $[\underline{t},\overline{t}]$. Then, construct from the data the representative chronometric function
\begin{equation}\label{cboundary} c^*(x) =\begin{cases} \underline{t} & \text{if } 1 < x,\\ (F^1)^{-1}(1-x) & \text{if } 0 < x \leq 1, \\ (F^0)^{-1}(1+x) & \text{if } -1 \leq x \leq 0,\\ \underline{t} & \text{if } x < -1, \\ \end{cases}\end{equation}
which reaches $c(x)=\underline{t}$ for finite absolute values of $x$. With this function, the empirical distribution $H$ becomes 
\[ H(x) =\begin{cases} 1 & \text{if } 1 < x,\\ 1-p^1+p^1x & \text{if } 0 < x \leq 1, \\ p^0+p^0x & \text{if } -1 \leq x \leq 0,\\ 0 & \text{if } x < -1, \\ \end{cases}\] which is piecewise linear. Applying any strictly sigmoid transformation to (\ref{cboundary}) generates a chronometric function that also attains $c(x)=\underline{t}$ for finite absolute values of $x$ but is piecewise more convex, resulting in a unimodal $H$. By the same logic, applying any inverse-sigmoid transformation generates a piecewise more concave chronometric function and a distribution $H$ that is not unimodal. The constructed function (\ref{cboundary}) therefore delimits sets of chronometric functions for which we can detect or reject unimodality. It follows from Theorem \ref{t1} that unimodality is detected for all those functions that are more convex than (\ref{cboundary}) and rejected for those that are more concave. This approach does not cover all possible chronometric functions but may yield expressive results. For example, if (\ref{cboundary}) plotted from the data is already strongly convex, then it appears unlikely that the true chronometric function is even more convex, and we may be able to reject the assumption of unimodality of the distribution.

\subsection{Correlation}

Our last application expands upon the baseline model by investigating the correlation between a latent variable $x \in \mathbb{R}$ (e.g., happiness, political attitude, or willingness to pay) and an observable variable represented by the index $j \in \mathbb{R}$ (e.g., income, social media usage, or experimental treatments). Such an application can be of value in the context of opinion surveys on topics that are sensitive and where subjects hesitate to provide certain answers that conflict with social norms \citep[e.g.,][]{coffman2017size}, which can result in limited variation in the response data. One may compensate for the lack of power of traditional analysis by replacing response variation with variation in response times. This can also be useful for experimental economists who are interested in the effect of different treatments on subjects' behavior. A null result in behavior does not necessarily establish the absence of an effect. If choices are discrete, such a null result can arise when the effect exists but is not large enough to shift behavior, given the parameters chosen by the experimenter. The effect may still be detectable in response times. The use of response time data for these purposes has already been proposed by \citet{konovalov2019revealed}.

We define a cumulative distribution function $\Gamma$ over the indices $j\in J\subseteq\mathbb{R}$, which is observable, and interpret $(G_j)_j$ as the conditional distributions of $x$ given each $j$. The joint distribution of $(x, j)$ is fully determined by $(G_j)_j$ and $\Gamma$. Further, just like in the baseline model, given observed data $(p_j, F_j)_j$ and a representative profile $(c^*_j)_j$ of chronometric functions, we can derive empirical distribution functions $(H_j)_j$ as defined in \eqref{defH}. These functions together with $\Gamma$ also give rise to a well-defined joint distribution of $(x, j)$.

A first attempt to quantify the association between  $x$ and $j$ is to employ the standard Pearson correlation coefficient
\begin{align}
\label{pearson-1}
    \rho=\frac{\mathit{Cov}(x,j)}{\sqrt{\mathit{Var}(x)\,\mathit{Var}(j)}},
\end{align}
where $\mathit{Cov}$ indicates the covariance function and $\mathit{Var}$ indicates the variance function. The value of this coefficient is invariant to all profiles $(\psi_j)_j$ of positive affine transformations that are identical for all $j\in J$. Consequently, when the chronometric functions are known up to identical linear transformations, we can ascertain the exact correlational pattern by computing \eqref{pearson-1} for the joint distribution of $(x, j)$ constructed via the functions $(H_j)_j$ and the marginal distribution $\Gamma$.

To circumvent the linearity restriction, one might opt for measuring the rank correlation between $x$ and $j$, which is particularly natural when $x$ and $j$ are ordinal variables \citep{kendall1955rank}. Intuitively, the rank of a variable is preserved under any monotone transformation, thereby allowing for more robust detection. For example, consider \citeauthor{spearman1904proof}'s \citeyearpar{spearman1904proof} rho, defined as
\[
    \rho_s=\frac{\mathit{Cov}(G(x),\Gamma(j))}{\sqrt{\mathit{Var}(G(x))\,\mathit{Var}(\Gamma(j))}},
\]
where the function $G$ represents the marginal distribution of $x$ and is given by
\begin{align*}
G(x)=\int_{J} G_j(x)\, d\Gamma(j) \text{ for all } x\in\mathbb{R}.  
\end{align*}
Since the rank of $x$ within the population remains unchanged under any profile $(\psi_j)_j$ of strictly increasing transformations that are identical across $j$, so does the value of Spearman's rho. A similar observation holds for another rank-based correlation measure, \citeauthor{kendall1955rank}'s \citeyearpar{kendall1955rank} tau, which is defined as
\[
    \rho_{\tau}=
    \mathbb{E}\left[\mathbbm{1}_{\{(x-x')(j-j')>0\}}\right]-\mathbb{E}\left[\mathbbm{1}_{\{(x-x')(j-j')<0\}}\right],
\]
where $(x', j')$ is distributed independently of $(x, j)$ but with the same joint distribution. Consequently, the rank correlation patterns for our variables of interest are detected under a fairly general class of chronometric functions, whenever they hold for the joint distribution induced by the appropriate functions $(H_j)_j$ and $\Gamma$.\footnote{The invariance property shared by Spearman's rho and Kendall's tau is related to the fact that they both depend only on the bivariate Copula of the two random variables, which is invariant to monotone transformations; see, e.g., \cite{fan2014copulas} and \cite{haugh2016}. The population versions of Spearman's rho and Kendall's tau that we adopt here were taken from \cite{haugh2016}.}

\end{appendix}

\end{document}